\documentclass[onecolumn, a4size, 11pt]{IEEEtran}

\usepackage{float}
\usepackage{amsmath,amsfonts,amssymb}
\usepackage{graphicx}
\usepackage{cite}
\usepackage{multirow}
\usepackage{algorithm}
\usepackage{algorithmic}
\usepackage{subfigure}
\usepackage{multirow}
\usepackage{url}
\usepackage{color}
\usepackage{array}
\usepackage{caption}
\usepackage{balance}

\graphicspath{{./figsJournalArxiv/}}

\IEEEoverridecommandlockouts

\makeatletter
\newcommand*{\rom}[1]{\expandafter\@slowromancap\romannumeral #1@}
\makeatother

\newtheorem{definition}{Definition}
\newtheorem{lemma}{Lemma}

\newcommand{\xE}{\mathbb{E}}

\newcommand{\Qnet}{\bar Q_{\text{net}}}

\begin{document}
\title{Optimized Training for Net Energy Maximization in Multi-Antenna Wireless Energy Transfer over Frequency-Selective Channel}
\author{Yong~Zeng and Rui~Zhang
\thanks{Y. Zeng is with the Department of Electrical and Computer Engineering, National University of Singapore, Singapore 117583 (e-mail: elezeng@nus.edu.sg).}
\thanks{R. Zhang is with the Department of Electrical and Computer Engineering, National University of Singapore, Singapore 117583 (e-mail: elezhang@nus.edu.sg). He is also with the Institute for Infocomm Research, A*STAR, Singapore 138632.}
}

\maketitle

\begin{abstract}
This paper studies the training design problem for multiple-input single-output (MISO) wireless energy transfer (WET) systems in frequency-selective channels, where the frequency-diversity and energy-beamforming gains can be both reaped  to maximize  the transferred energy by efficiently  learning the channel state information (CSI) at the energy transmitter (ET). By exploiting   channel reciprocity, a new two-phase  channel training scheme is proposed to achieve the diversity and beamforming gains, respectively.  In the first phase, pilot signals are sent from the energy receiver (ER) over a selected subset of the available frequency sub-bands, through which the ET determines a certain number of  ``strongest'' sub-bands with largest antenna sum-power gains and sends their indices to the ER. In the second phase, the selected sub-bands are further trained by the ER, so that the ET obtains a refined  estimate of the corresponding MISO channels to implement energy beamforming for WET. A training design problem is formulated and optimally solved, which takes into account the channel training overhead by maximizing the \emph{net} harvested energy at the ER, defined as the average harvested energy offset by that consumed in the two-phase training. Moreover, asymptotic analysis is obtained  for systems with a large number of antennas or a large number of sub-bands to gain useful insights on the optimal training design. Finally, numerical results are provided to corroborate our analysis and show the effectiveness of the proposed scheme that optimally balances the diversity and beamforming gains achieved in MISO WET systems  with limited-energy training.
\end{abstract}


\section{Introduction}
 Wireless energy transfer (WET), by which energy is delivered wirelessly from one location  to another without the hassle of connecting cables, has emerged as a promising solution to provide convenient and reliable power supplies for energy-constrained wireless networks, such as radio-frequency identification (RFID) and sensor networks \cite{525,534}. One enabling technique of WET for long-range applications (say up to tens of meters) is via radio-frequency (RF) or microwave propagation, where dedicated energy-bearing signals are sent from the energy transmitter (ET) over a particular frequency band to the energy receiver (ER) for harvesting the received RF energy (see e.g. \cite{525,534} and the references therein).


By exploiting the fact that RF signal is able to convey both information and energy, extensive studies have been recently devoted to unifying both information and energy transfer under the framework so-called simultaneous wireless information and power transfer (SWIPT) \cite{478,514,521}. The advancement of WET techniques has also opened several interesting applications for SWIPT and other wireless-powered communication networks (WPCN) \cite{525,515}, such as cognitive radio networks \cite{510,509} and relay networks \cite{535,511,551,541}. The key practical challenge for implementing  WET is to improve its energy transfer efficiency, which is mainly limited  by the significant power attenuation loss over distance. In line-of-sight (LOS) environment, the efficiency can be enhanced via uni-directional transmission, by designing the antenna radiation patterns to form a single sharp beam towards the ER for energy focusing. Recently, multi-antenna based technique, termed {\it energy beamforming}, has been proposed for WET \cite{478}, where multiple antennas are employed at the ET to more efficiently and flexibly direct wireless energy to one or more ERs via digital beamforming, especially in non-LOS environment with multiple signal paths from the ET to ERs. Similar to the emerging massive multiple-input multiple-output (MIMO) enabled wireless communications (see e.g. \cite{374,497} and the references therein), by equipping a large number of antennas at the ET, enormous  energy beamforming gain can be achieved; hence, the end-to-end energy transfer efficiency can be dramatically enhanced \cite{552}.

  On the other hand, for WET in a wide-band regime over frequency-selective channels, the frequency-diversity gain can also be exploited to further improve the energy transfer efficiency, by transmitting more power over the sub-band with higher channel gain. Ideally, maximum energy transfer efficiency is attained by sending just one sinusoid at the frequency that exhibits the strongest channel frequency response \cite{504}. In practice, it may be necessary to spread the transmission power over multiple strong sub-bands in order to comply with various regulations, such as the power spectral density constraint imposed on the license-free ISM (industrial, scientific and medical) band \cite{549}. WET in single- \cite{504,522,526} and multi-antenna \cite{527} frequency-selective channels have been studied under the more general SWIPT setup, where perfect channel state information (CSI) is assumed at the transmitter.

   In practice, both the energy-beamforming and frequency-diversity gains can be exploited for multi-antenna WET in frequency-selective channels; while they crucially depend on the available CSI at the ET, which needs to be practically obtained  at the cost of additional time and energy consumed. Similar to wireless communication,  a direct approach to obtain CSI is by sending pilot signals from the ET to the ER \cite{495,500},  which estimates the corresponding channel and then sends the estimation back to the ET via a finite-rate feedback link  (see e.g. \cite{484} and the references therein). However, since the training overhead increases with the number of antennas $M$ at the ET, this method is not suitable when $M$ is large \cite{373}. Furthermore, as pointed out in \cite{491}, estimating the channel at the ER requires complex baseband signal processing, which may not be available at the ER due to its practical hardware limitation. A new channel-learning design to cater for the practical RF energy harvesting circuitry at the ER has thus been proposed in \cite{491}. However, the training overhead still increases quickly with $M$, and can be prohibitive as $M$ becomes large. In \cite{528}, by exploiting channel reciprocity between the forward (from the ET to the ER) and reverse (from the ER to the ET) links, we have proposed an alternative channel-learning scheme for WET based on reverse-link training by sending pilot signals from the ER to the ET, which is more efficient since the training overhead is independent of $M$. Moreover, it is revealed in \cite{528} that a non-trivial trade-off exists in WET systems between maximizing the energy beamforming gain by more accurately estimating the channel at the ET and minimizing the training overhead, including the time and energy used for sending the pilot signals from the ER. To optimally resolve the above trade-off, a training design framework has been proposed in \cite{528} based on maximizing the {\it net} harvested energy at the ER, which is the average harvested energy offset by that used for channel training. However, the proposed design in \cite{528} applies only for narrowband flat-fading channels instead of the more complex wide-band frequency-selective fading channels, which motivates the current work.

  In this paper, we study the channel learning  problem for multiple-input single-output (MISO) point-to-point WET systems over frequency-selective fading channels. Compared to its narrowband counterpart studied in \cite{528}, the channel frequency selectivity in MISO wide-band WET system provides  an additional frequency-diversity gain for further  improvement of the  energy transfer efficiency \cite{557}. However, the training design also becomes more challenging, since the channels in both the frequency and spatial domains need to be estimated. In particular, achieving the frequency-diversity gain requires only the knowledge of the {\it total} channel power from all the $M$ transmit antennas at the ET, but for a potentially large number of sub-bands; while the energy beamforming gain relies on a more accurate estimation of the MISO channel (both phase and magnitude for each antenna), but only for a small subset of sub-bands selected with largest channel power for energy transmission. Therefore, an efficient training design for multi-antenna wide-band WET needs to achieve a good balance between exploiting the frequency-diversity gain (by training more independent sub-bands) versus the  energy beamforming gain (by allocating more training energy to the selected sub-bands), and yet without consuming excessive energy at the ER, as will be pursued in this paper. Specifically, the main contributions of this paper are summarized as follows:

  \begin{itemize}
  \item First, under channel reciprocity, we propose a new two-phase channel training scheme for MISO wide-band WET, to exploit the frequency-diversity and energy beamforming gains, respectively. In phase-\rom{1} training, pilot signals are sent from the ER over a selected subset of the available frequency sub-bands. Based on the received total energy across all the antennas at the ET over each of the trained sub-bands, the ET determines a certain number of  sub-bands with largest antenna sum-power gains and sends their ordered indices to the ER. In phase-\rom{2} training, the selected sub-bands are further trained by the ER, so that the ET obtains a refined estimate of their MISO channels to implement energy beamforming over them for WET. The proposed two-phase training protocol is more effective than the single-phase training schemes with phase-\rom{1} or phase-\rom{2} training only, or the brute-force training scheme by estimating the (exact) MISO channels at all the available sub-bands.
  \item Second, based on the order statistic analysis \cite{546},  a closed-form expression is derived for the average harvested energy at the ER with the proposed two-phase training scheme, in terms of the number of sub-bands trained, as well as the pilot signal energy allocated in the two training phases. An optimization problem is thus formulated to maximize the {\it net} harvested energy at the ER, which explicitly takes into account the energy consumption at the ER for sending the pilot signals. The formulated problem is optimally solved based on convex optimization techniques to obtain the optimal training design.
  \item Third, we provide the asymptotic analysis for systems with a large number of antennas (massive-array WET) or a large number of independent sub-bands (massive-band WET) to gain further  insights on the optimal training design. It is shown that while the {\it net} harvested energy at the ER linearly increases with the number of ET antennas asymptotically, it is upper-bounded by a constant value as the number of independent sub-bands goes to infinity. The latter result is a consequence of the faster increase in training energy cost (linearly) than that in the achieved frequency-diversity gain for WET (logarithmically) as more independent sub-bands are trained.
  \end{itemize}

The rest of this paper is organized as follows. Section~\ref{sec:systemModel} introduces the system model. Section~\ref{sec:formulation} discusses the proposed two-phase training scheme for MISO wide-band WET and formulates the training optimization problem. In Section~\ref{sec:optimalTraining}, the formulated problem is optimally solved to obtain the optimal training strategy. In Section~\ref{sec:asympt}, asymptotic analysis is provided for massive-array and massive-band WET systems, respectively. Numerical results are provided in Section~\ref{sec:simulation} to corroborate our study. Finally, we conclude the paper in Section~\ref{sec:conclusion}.

\emph{Notations:} In this paper, scalars are denoted by italic letters. Boldface lower- and upper-case letters denote vectors and matrices, respectively. $\mathbb{C}^{M\times N}$ denotes the space of $M\times N$ complex matrices. $\mathbb{E}[\cdot]$ denotes the expectation operator.  $\mathbf{I}_M$ denotes an $M$-dimensional identity matrix and $\mathbf{0}$ denotes an all-zero matrix of appropriate dimension. For an arbitrary-size matrix $\mathbf{A}$,  its complex conjugate, transpose, and Hermitian transpose are respectively denoted as $\mathbf A^*$, $\mathbf{A}^{T}$ and $\mathbf{A}^{H}$. 
  $\|\mathbf x\|$ denotes the Euclidean norm of vector $\mathbf x$. 
   For a real number $x$, $[x]^+\triangleq \max\{x,0\}$. 
    $a!$ denotes the factorial of $a$ and $\binom{n}{k}$ represents the binomial coefficient. Furthermore, $\mathcal{CN}(\boldsymbol \mu, \mathbf C)$ represents the circularly symmetric complex Gaussian (CSCG) distribution with mean $\boldsymbol \mu$ and covariance matrix $\mathbf C$.  Finally, for two positive-valued functions $f(x)$ and $g(x)$, the notation $f(x)=\Theta(g(x))$ implies that $\exists c_1, c_2>0$ such that $c_1g(x)\leq f(x) \leq c_2 g(x)$ for sufficiently large $x$. 

 \section{System Model}\label{sec:systemModel}
 \subsection{MISO Multi-Band WET}
We consider a MISO point-to-point WET system in frequency-selective channel, where an ET equipped with $M\geq 1$ antennas is employed to deliver wireless energy to a single-antenna ER. We assume that the total available bandwidth is $B$, which is equally divided into $N$ orthogonal sub-bands with the $n$th sub-band centered at  frequency $f_n$ and of bandwidth $B_s=B/N$. We assume that $B_s< B_c$, where $B_c$ denotes the channel coherence bandwidth, so that the channel between the ET and ER experiences frequency-flat fading within each sub-band. Denote $\mathbf h_n\in \mathbb{C}^{M\times 1}$, $n=1,\cdots, N$, as the complex-conjugated baseband equivalent MISO channel from the ET to the ER in the $n$th sub-band. We assume a quasi-static  Rayleigh fading model, where $\mathbf h_n$ remains constant within each block of $T< T_c$ seconds, with $T_c$ denoting the channel coherence time, but can vary from one block to another. 
 Furthermore, the elements in $\mathbf h_n$ are modeled as independent and identically distributed (i.i.d.) zero-mean CSCG random variables with  variance $\beta$, i.e.,
\begin{align}
\mathbf h_n \sim \mathcal{CN}(\mathbf 0, \beta \mathbf I_M), \ n=1,\cdots, N, \label{eq:hn}
\end{align}
where $\beta$ models the large-scale fading  due to shadowing as well as the distance-dependent path loss.

\begin{figure}
\centering
\includegraphics[scale=0.7]{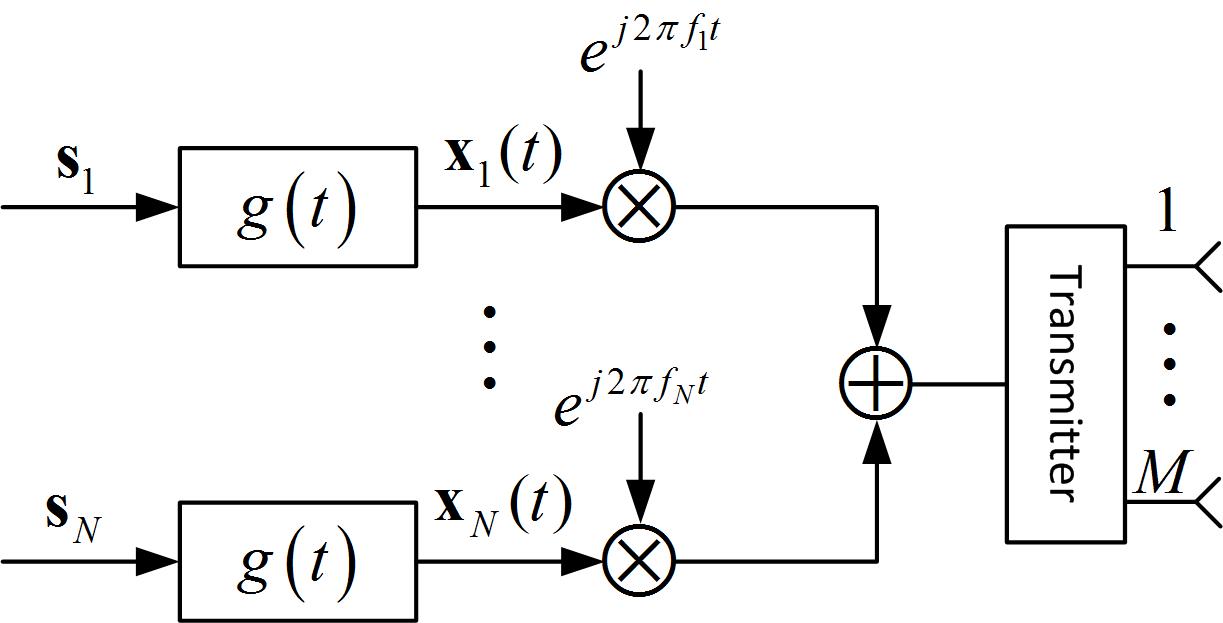}
\caption{Schematics of the  multi-antenna multi-band wireless energy transmitter.}\label{F:transmitter}
\end{figure}

Within each block of $T$ seconds, i.e., $0\leq t\leq T$, the input-output relation for the forward link energy transmission can be expressed as
\begin{align}
y_n(t)= \mathbf h_n^H \mathbf x_n(t) + z_n(t), \ n=1,\cdots, N, \label{eq:yt}
\end{align}
where $y_n(t)$ denotes the received signal at the ER in the $n$th sub-band;
 $\mathbf x_n(t)\in \mathbb{C}^{M\times 1}$ denotes the baseband energy-bearing signals transmitted by the ET in sub-band $n$; and $z_n(t)$ denotes the additive noise at the ER. Different from wireless communication where random signals need to be transmitted to convey information, $\mathbf x_n(t)$ in \eqref{eq:yt} is designated  only for energy transmission and thus can be chosen to be  deterministic. The average power $p_n$ transmitted over sub-band $n$ can then be expressed as
\begin{align}
p_n= \frac{1}{T}\int_0^T \left \| \mathbf x_n(t)\right \|^2 dt, \ n=1,\cdots, N. \label{eq:pn}
\end{align}

At the ER, the incident RF power captured by the antenna is converted to direct current (DC) power for battery replenishment by a device called rectifier \cite{514}. We assume that the RF energy harvesting circuitry is able to achieve efficient energy harvesting over the entire operation frequency band. 
By ignoring the energy harvested from the background noise which is practically small, the total harvested energy at the ER over all $N$ sub-bands during each block can be expressed as \cite{478}
\begin{equation} \label{eq:Q}
 \begin{aligned}
  Q=\eta \sum_{n=1}^N \int_0^T \left|\mathbf h_n^H \mathbf x_n(t)\right|^2dt,
  \end{aligned}
 \end{equation}
 where $0<\eta\leq 1$ denotes the energy conversion efficiency at the ER. 
  Without loss of generality, $\mathbf x_n(t)$ can be expressed as (see Fig.~\ref{F:transmitter} for the transmitter schematics) 
 \begin{align}
 \mathbf x_n(t)=\mathbf s_n  g(t), \ 0\leq t\leq T, \ n=1,\cdots, N, \label{eq:xn}
 \end{align}
 where $\mathbf s_n \in \mathbb{C}^{M\times 1}$ denotes the energy beamforming vector for the $n$th sub-band, and $g(t)$ represents the pulse-shaping waveform with normalized power, i.e., $\frac{1}{T} \int_0^T |g(t)|^2dt=1$. 
   Note that the bandwidth of $g(t)$, which is approximately equal to $1/T$, needs to be no larger than $B_s$. We thus have
 \begin{align}
 \frac{1}{T_c}< \frac{1}{T}<B_s < B_c,  \notag 
 \end{align}
 or $T_cB_c> 1$, i.e., a so-called ``under-spread'' wide-band fading channel is assumed \cite{76}.

\subsection{Optimal WET with Perfect CSI}\label{sec:perfectCSI}
  With  \eqref{eq:pn} and \eqref{eq:xn}, the power $p_n$ transmitted over sub-band $n$ can be expressed as $p_n=\|\mathbf s_n\|^2$, $\forall n$. 
  Furthermore, the harvested energy $Q$ in \eqref{eq:Q} can be written as $Q=\eta T \sum_{n=1}^N |\mathbf h_n^H \mathbf s_n|^2$. In the ideal case with perfect CSI, i.e., $\{\mathbf h_n\}_{n=1}^N$, at the ET, the optimal design of $\{\mathbf s_n\}_{n=1}^N$  that maximizes $Q$ can be obtained by solving the following problem
\begin{equation}\label{P:perfectCSI}
\begin{aligned}
 \quad \max \ & \eta T \sum_{n=1}^N \left|\mathbf h_n^H \mathbf s_n \right|^2 \\
\text{ subject to } & \sum_{n=1}^N \| \mathbf s_n\|^2\leq P_t, \\
& \|\mathbf s_n\|^2 \leq P_s, \ \forall n,
\end{aligned}
\end{equation}
where $P_t$ denotes the maximum transmission power at the ET across all the $N$ sub-bands, and $P_s$ corresponds to the power spectrum density constraint for each sub-band \cite{549}. Without loss of generality, we assume that $P_s\leq P_t \leq N P_s$, since otherwise at least one of the constraint in \eqref{P:perfectCSI} is redundant and hence can be removed. In addition, for the convenience of exposition, we assume that $P_t$ is a multiple of $P_s$, i.e., $P_t/P_s=N_2$ for some integer $1\leq N_2\leq N$.
It can then be obtained that the optimal solution to problem \eqref{P:perfectCSI} is
\begin{align}\label{eq:optxn}
\mathbf s_{[n]}= \begin{cases}
\sqrt {P_s} \frac{\mathbf h_{[n]}}{\|\mathbf h_{[n]}\|}, & n=1,\cdots, N_2, \\
\mathbf 0, & n=N_2+1,\cdots, N,
\end{cases}
\end{align}
where $[\cdot]$ is a permutation such that $\|\mathbf h_{[1]}\|^2 \geq \|\mathbf h_{[2]}\|^2 \cdots \geq \|\mathbf h_{[N]}\|^2$.
The resulting harvested energy is
\begin{align}\label{eq:Qmax}
Q^{\text{ideal}}=\eta T P_s \sum_{n=1}^{N_2} \|\mathbf h_{[n]}\|^2.
\end{align}
It is observed from \eqref{eq:optxn} that for a MISO multi-band WET system with $P_t<NP_s$, or $N_2<N$, the optimal scheme is to transmit over the  $N_2$ strongest sub-bands only, each with the maximum allowable power. As a result, the remaining $N-N_2$ unused sub-bands  can be opportunistically re-used for other applications such as information transmission.  The solution in \eqref{eq:optxn} also shows  that for each of the $N_2$ selected sub-bands, maximum ratio transmission (MRT) should be performed across different transmit antennas to achieve the maximum energy beamforming gain. Moreover, \eqref{eq:Qmax} implies that for multi-antenna  WET systems in frequency-selective channels, both frequency-diversity and energy beamforming gains can be achieved to maximize the energy transfer efficiency in the ideal case with perfect CSI at the ET.

\subsection{Proposed Two-Phase Training Protocol}
In practice, the CSI needs to be obtained  at the ET via channel training and/or feedback from the ER, which are subject to channel estimation and/or feedback errors and incur additional energy and time costs. Consequently, the maximum harvested energy given in \eqref{eq:Qmax} cannot be attained; instead, it only provides a performance upper bound for practical WET systems, for which efficient channel-learning schemes need to be designed. In this paper, by exploiting the channel reciprocity that the forward and reverse link channels are transpose of each other, we propose a two-phase channel training scheme to reap both  the frequency-diversity and energy beamforming gains. As illustrated in Fig.~\ref{F:twoPhaseTraining}, the first  phase corresponds to the first $\tau_1< T$ seconds of each block, where pilot signals are sent by the ER to the ET over $N_1$ out of the $N$ available sub-bands, each with energy $E_1$, where $N_2\leq N_1\leq N$. By comparing the total received energy from all $M$ antennas at the ET over each of the $N_1$ trained sub-bands,  the ET determines the $N_2$ strongest sub-bands with largest antenna sum-power gains, and sends their corresponding orders and indices to the ER. In the second phase of $\tau_2< T-\tau_1$ seconds, additional training signal is sent by the ER over the $N_2$ selected sub-bands, each with possibly different training energy depending on their relative order as determined in phase \rom{1}. The ET then obtains an estimate of the  MISO channel for each of the $N_2$ selected sub-bands, based on which MRT beamforming is performed for energy transmission during the remaining $T-\tau_1-\tau_2$ seconds. The proposed two-phase training scheme is elaborated in more details in the next section.

\begin{figure}
\centering
\includegraphics[scale=0.8]{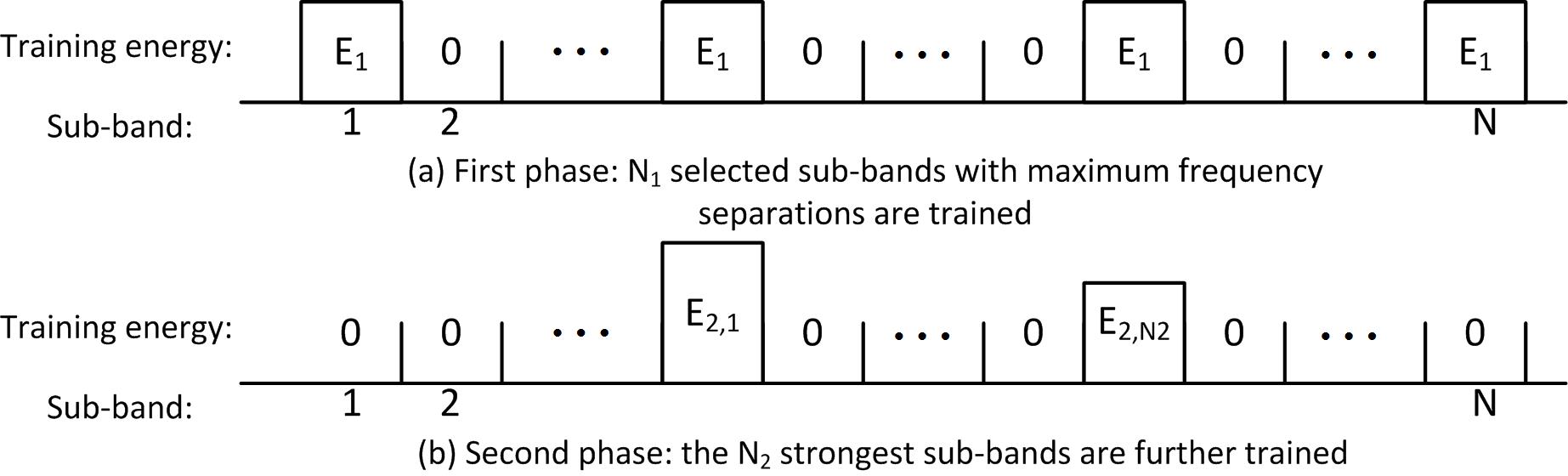}
\caption{Two-phase channel  training for multi-antenna multi-band wireless energy transfer.}\label{F:twoPhaseTraining}
\end{figure}

\section{Problem Formulation}\label{sec:formulation}
\subsection{Two-Phase Training}
\subsubsection{Phase-\rom{1} Training}
 Denote by $\mathcal{N}_1\subset \{1,\cdots, N\}$ the set of the $N_1$ trained sub-bands in phase \rom{1}, with $N_2\leq N_1\leq N$.  To maximize the frequency-diversity gain, the sub-bands with the maximum frequency separations are selected  in $\mathcal{N}_1$ so that their channels are most likely to be independent (see Fig.~\ref{F:twoPhaseTraining}(a)), e.g., if $N_1=2$, we have $\mathcal{N}_1=\{1,N\}$.  The received training signals at the ET can be written as
\begin{align}
\mathbf r_n^{\text{\rom{1}}}(t)=\sqrt{E_1} \mathbf h_n^* \phi_n(t) + \mathbf w_n^{\text{\rom{1}}}(t), \ 0\leq t \leq \tau_1, \ n\in  \mathcal{N}_1,
\end{align}
where 
$E_1$ denotes the training energy used by the ER for each trained sub-band; $\phi_n(t)$ represents the training waveform for sub-band $n$ with normalized energy, i.e., $\int_0^{\tau_1} |\phi_n(t)|^2dt=1$, $\forall n$;  and $\mathbf w_n^{\text{\rom{1}}}(t)\in \mathbb{C}^{M\times 1}$ represents the additive white Gaussian noise (AWGN) at the receiver of the ET with power spectrum density $N_0$.  The total energy consumed at the ER for channel training in this phase  is
\begin{align}
E_{\text{tr}}^{\text{\rom{1}}}=\sum_{n\in \mathcal{N}_1} \int_0^{\tau_1}\left|\sqrt{E_1} \phi_n(t)\right|^2=E_1 N_1.
\end{align}

At the ET, the received training signal is first separated over different sub-bands. After complex-conjugate operation, each $\mathbf r_n^{\text{\rom{1}}*}(t)$ then passes through a matched filter to obtain
\begin{align}
\mathbf y_n^{\text{\rom{1}}}= \int_0^{\tau_1} \mathbf r_n^{\text{\rom{1}}*}(t) \phi_n(t) dt=\sqrt{E_1} \mathbf h_n + \mathbf z_n^{\text{\rom{1}}}, \ n\in \mathcal{N}_1, \label{eq:ynI}
\end{align}
where $\mathbf z_n^{\text{\rom{1}}}\sim \mathcal{CN}(\mathbf 0, N_0 \mathbf I_M)$ denotes the i.i.d. AWGN vector.
Based on \eqref{eq:ynI}, the ET determines the $N_2$ out of the $N_1$ trained sub-bands that have the largest power gains. Specifically, the $N_1$ sub-bands in $\mathcal{N}_1$ are ordered based on the total received energy across all the $M$ antennas such that
\begin{align}
\|\mathbf y_{[1]}^{\text{\rom{1}}}\|^2 \geq \|\mathbf y_{[2]}^{\text{\rom{1}}}\|^2 \geq \cdots \geq \|\mathbf y_{[N_1]}^{\text{\rom{1}}}\|^2. \label{eq:order}
\end{align}

The ET then sends to the ER the {\it ordered} indices of the $N_2$ strongest sub-bands $\mathcal{N}_2=\{[1],\cdots, [N_2]\}$. 

 \begin{figure}
 \vspace{-3ex}
\centering
\includegraphics[scale=0.25]{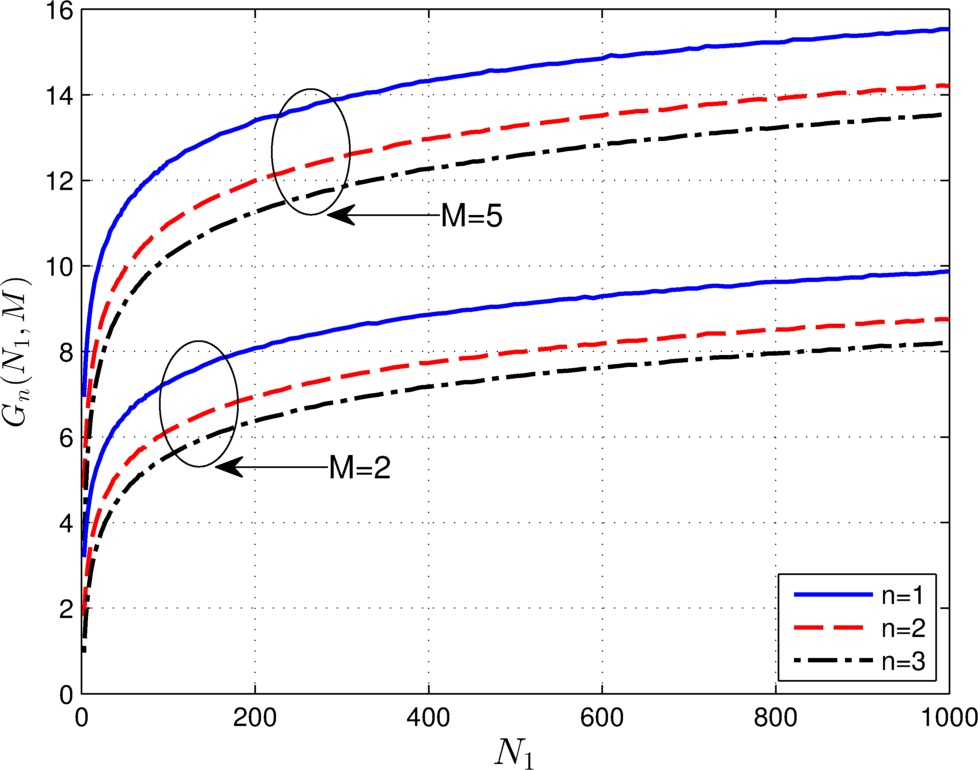} 
\caption{$G_n(N_1,M)$, $n=1,2,3$, versus $N_1$ for $M=2$ and $M=5$.}\label{F:GVsN1}
\end{figure}

\subsubsection{Phase-\rom{2} Training}
In the second phase of $\tau_2$ seconds, additional pilot signals are transmitted by the ER over the $N_2$ selected sub-bands for the ET to estimate the corresponding MISO channels.  Denote by $E_{2,n}$ the training energy used for the $n$th strongest sub-band as determined in \eqref{eq:order}, $n=1,\cdots, N_2$. With similar processing as in phase \rom{1}, the resulting  signal at the ET during phase \rom{2} over sub-band $[n]$ can be expressed as
%
\begin{align}
\mathbf y_{[n]}^{\text{\rom{2}}} = \sqrt{E_{2,n}} \mathbf h_{[n]} + \mathbf z_{[n]}^{\text{\rom{2}}},  \ n=1,\cdots, N_2,\label{eq:ynII}
\end{align}
where $\mathbf z_{[n]}^{\text{\rom{2}}}\sim \mathcal{CN}(\mathbf 0, N_0 \mathbf I_M)$ denotes the AWGN.
The ET then performs the linear minimum mean-square error (LMMSE) estimation for $\mathbf h_{[n]}$  based on $\mathbf y_{[n]}^{\text{\rom{2}}}$. \footnote{In principle, $\mathbf h_{[n]}$ can be estimated based on both observations $\mathbf y_{[n]}^{\text{\rom{1}}}$ and $\mathbf y_{[n]}^{\text{\rom{2}}}$. To simplify the processing of multi-band energy detection in phase-\rom{1} training, we assume that $\mathbf y_{[n]}^{\text{\rom{1}}}$ is only used for determining the strongest sub-bands via estimating the power $\|\mathbf h_{[n]}\|^2$ while  only $\mathbf y_{[n]}^{\text{\rom{2}}}$ is used for estimating the exact MISO channel $\mathbf h_{[n]}$.}
 To obtain the optimal LMMSE estimator, we first provide  the following lemma.

\begin{lemma}\label{lemma:exphnSq}
Given that $\mathbf h_n$ and $\mathbf h_m$ are independent $\forall n,m\in \mathcal{N}_1$ and $n\neq m$, the average  power of $\mathbf h_{[n]}$, where $[n]$ denotes the $n$th strongest sub-band determined via phase-\rom{1} training as in \eqref{eq:order},  can be expressed as
\begin{align}
R_n(N_1,E_1)\triangleq \xE\left [\left \| \mathbf h_{[n]} \right\|^2 \right] & =\frac{\beta^2 E_1 G_n(N_1,M)+\beta N_0 M}{\beta E_1 + N_0},\notag \\
& n=1,\cdots, N_2,\label{eq:Ehnsq}
\end{align}
with $G_n(N_1,M)$ an increasing function of $N_1$ and $M$ as defined in \eqref{eq:G1} and \eqref{eq:Gnplus1} in Appendix~\ref{A:usefulLemma}, which satisfies:
 \begin{align}
 G_1(N_1,M)\geq G_2(N_1,M)\geq \cdots \geq G_{N_2}(N_1,M) \geq M. \label{eq:Gorder}
\end{align}
\end{lemma}

\begin{IEEEproof}
The proof of Lemma~\ref{lemma:exphnSq} requires some basic knowledge of order statistic \cite{546}, which is briefly introduced in Appendix~\ref{A:orderStatistic}. The detailed proof is given in Appendix~\ref{A:exphnSq}, following a useful lemma presented in Appendix~\ref{A:usefulLemma}.
\end{IEEEproof}

\begin{figure}
\centering
\includegraphics[scale=0.25]{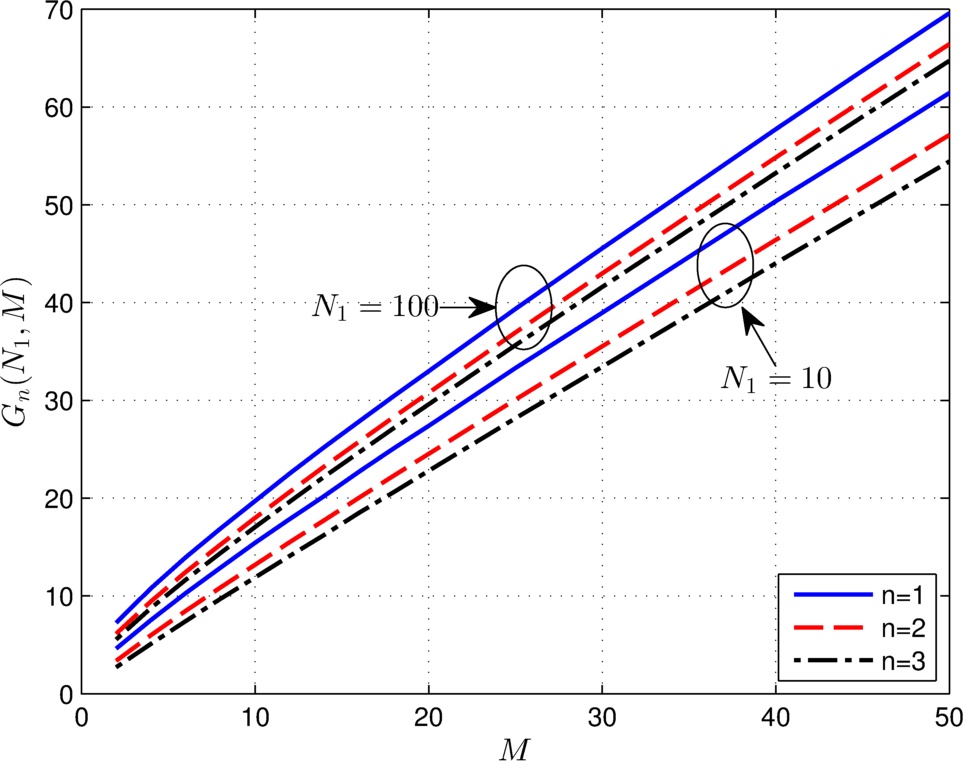}
\caption{$G_n(N_1,M)$, $n=1,2,3$, versus $M$ for $N_1=10$ and $N_1=100$.}\label{F:GVsM}
\end{figure}

It can be inferred from Lemma~\ref{lemma:exphnSq}  that $R_n(N_1,E_1)$ monotonically increases with the training-energy-to-noise ratio $E_1/N_0$. In the extreme case when $E_1/N_0\rightarrow 0$, for which the energy comparison in \eqref{eq:order} is solely determined by the additive noise and hence is unreliable,  $R_n(N_1,E_1)$ in \eqref{eq:Ehnsq} reduces to $\beta M$, $\forall n$, consistent with  the assumed channel statistics in \eqref{eq:hn}. In this case, no frequency-diversity gain is achieved. On the other hand, as $E_1/N_0\rightarrow \infty$, $R_n(N_1,E_1)$ approaches to its upper bound $\beta G_n(N_1,M)$. Thus, the ratio $G_n(N_1,M)/M$ can be interpreted as the maximum frequency-diversity gain achievable for the $n$th strongest sub-channel by training  $N_1$ independent sub-bands. As an illustration, Fig.~\ref{F:GVsN1} and Fig.~\ref{F:GVsM} plot the numerical values of $G_n(N_1,M)$ for $n=1,2,3$ against $N_1$ and $M$, respectively. It is observed that while $G_n(N_1,M)$ increases linearly with $M$, it scales with $N_1$ only in a logarithmic manner.

  With Lemma~\ref{lemma:exphnSq}, the LMMSE estimator of the MISO channels corresponding to the $N_2$  selected sub-bands in phase-\rom{2} training   can be obtained as follows.
\begin{lemma}\label{lemma:LMMSE}
The LMMSE estimator $\hat {\mathbf h}_{[n]}$ for the MISO channel $\mathbf h_{[n]}$  based on \eqref{eq:ynII} is
\begin{align}
\hat {\mathbf h}_{[n]}=\frac{\sqrt{E_{2,n}}R_n(N_1,E_1)}{E_{2,n}R_n(N_1,E_1)+N_0 M}\mathbf y_{[n]}^{\text{\rom{2}}}, \ n=1,\cdots, N_2. \label{eq:hnLMMSE}
\end{align}
Define the channel estimation error as $\tilde{\mathbf h}_{[n]} \triangleq \mathbf h_{[n]}-\hat {\mathbf h}_{[n]}$. We then have
\begin{align}
&\xE\left[\|\tilde{\mathbf h}_{[n]}\|^2 \right]=\frac{N_0MR_n(N_1,E_1) }{E_{2,n} R_n(N_1,E_1) + N_0 M}, \label{eq:mse}\\
&\xE\left[\|\hat {\mathbf h}_{[n]}\|^2 \right]=\frac{E_{2,n} R_n^2(N_1,E_1)}{E_{2,n}R_n(N_1,E_1)+N_0M},\label{eq:ehhat}\\
&\xE\left[ \tilde {\mathbf h}_{[n]}^H\hat {\mathbf h}_{[n]} \right]=0, \ n=1,\cdots, N_2.\label{eq:ecross}
\end{align}
\end{lemma}
\begin{IEEEproof}
Please refer to Appendix~\ref{A:LMMSE}.
\end{IEEEproof}

\subsection{Net Harvested Energy Maximization}
After  the two-phase training, energy beamforming is performed by the ET over the $N_2$ selected sub-bands based on the estimated MISO channels $\{\hat{\mathbf h}_{[n]}\}_{n=1}^{N_2}$ for WET during the remaining time of $T-\tau_1-\tau_2$ seconds. According to \eqref{eq:optxn}, the energy beamforming vector for the $n$th selected sub-band is set as  $\mathbf s_{[n]}=\sqrt{P_s}\hat{\mathbf h}_{[n]}/\|\hat{\mathbf h}_{[n]}\|$, $n=1,\cdots, N_2$.
The resulting  energy harvested at the ER can be expressed as\footnote{ For notational simplicity, we assume that $T$ is sufficiently large so that $T\gg \tau_1+\tau_2$; as a result, the time overhead for channel training is ignored (but energy cost of channel training remains).If the above assumption does not hold, the results derived in this paper are still valid by simply replacing $T$ with $T-\tau_1-\tau_2$ with any given $\tau_1$ and $\tau_2$.}
\begin{align}
\hat Q &= \eta T P_s \sum_{n=1}^{N_2} \frac{\left|\mathbf h_{[n]}^H \hat{\mathbf h}_{[n]}\right|^2}{\|\hat {\mathbf h}_{[n]}\|^2} \notag \\
&=\eta T P_s \sum_{n=1}^{N_2}\Big( \|\hat{\mathbf h}_{[n]} \|^2 + \frac{|\tilde{\mathbf h}_{[n]}^H \hat{\mathbf h}_{[n]}|^2}{\|\hat {\mathbf h}_{[n]}\|^2}
+\tilde{\mathbf h}_{[n]}^H \hat{\mathbf h}_{[n]} + \hat{\mathbf h}_{[n]}^H\tilde{\mathbf h}_{[n]}\Big),\label{eq:hatQ2}
\end{align}
where we have used the identity $\mathbf h_{[n]}=\hat{\mathbf h}_{[n]}+ \tilde{\mathbf h}_{[n]}$ in \eqref{eq:hatQ2}. Based on \eqref{eq:mse}-\eqref{eq:ecross}, the average harvested energy at the ER can be expressed  as
\begin{align}
\bar Q & (N_1,E_1,\{E_{2,n}\}) = \xE\left[ \hat Q\right] \notag \\ 
&=\eta T P_s \sum_{n=1}^{N_2} R_n(N_1,E_1)\Big(1-\frac{(M-1)N_0}{E_{2,n}R_n(N_1,E_1)+N_0M}\Big). \label{eq:barQ}
\end{align}

It is observed from \eqref{eq:barQ} that for each of the $N_2$ selected sub-bands, the average harvested energy is given by a difference of two terms. The first term, $\eta T P_s R_n(N_1,E_1)$, is the average harvested energy when energy beamforming is based on the perfect knowledge of $\mathbf h_{[n]}$ over sub-band $[n]$, where $[n]$ is the $n$th strongest sub-band determined  via phase-\rom{1} training among the $N_1$ trained sub-bands  each with training energy $E_1$. 
 The second term can be interpreted as the loss in energy beamforming performance due to the error in the  estimated MISO channel $\hat{\mathbf h}_{[n]}$ in phase-\rom{2} training. For the extreme scenario with $E_{2,n}/N_0 \rightarrow \infty$ so that $\mathbf h_{[n]}$ is perfectly estimated, or in SISO system ($M=1$) for which no energy beamforming can be applied, the second term in \eqref{eq:barQ} vanishes, as expected. 

 The \emph{net} harvested energy at the ER, which is the average harvested energy offset by  that used for sending pilot signals in the two-phase training, can be expressed as 
\begin{align}\label{eq:Qnet}
 \hspace{-2ex} \Qnet & (N_1,E_1,\{E_{2,n}\}) =\bar Q(N_1,E_1,\{E_{2,n}\})- E_1N_1-\sum_{n=1}^{N_2}E_{2,n}.
\end{align}
 The problem of finding the optimal training design to maximize  $\Qnet$ can be formulated as
\begin{align}
\mathrm{(P1):} \qquad \underset{ N_1, E_1, \{E_{2,n}\}}{\max}   \quad  & \Qnet  (N_1,E_1,\{E_{2,n}\}) \notag \\
\text{subject to} \quad &  N_2\leq N_1 \leq N, \notag \\
& E_1\geq 0,\  E_{2,n}\geq 0, \ n=1,\cdots, N_2. \notag
\end{align}
Note that in (P1), we have assumed that there is sufficiently large initial energy stored at the ER in the first block for sending training signals. Under this assumption and considering the fact that our optimized net energy $\Qnet$ in (P1) is strictly positive and generally much larger than that required for training (for implementing other more energy-consuming functions at the ER such as sensing and data transmission), we further assume that there is always enough energy available to send the training signals at the beginning of each subsequent block.

In the following sections, we first develop an efficient algorithm to obtain the optimal solution to (P1), and then investigate the problem in the asymptotic regime with large $M$ or large $N$ to gain useful insights.

\section{Optimal Training Design}\label{sec:optimalTraining}
To find the optimal solution to (P1), we first obtain the optimal training energy $\{E_{2,n}\}_{n=1}^{N_2}$ in phase \rom{2} with $N_1$ and $E_1$ fixed. In this case, (P1) can be decoupled into $N_2$ parallel sub-problems. By discarding constant terms, the $n$th sub-problem can be formulated as
\begin{equation}\label{P:solveE2}
\begin{aligned}
\underset{E_{2,n}\geq 0}{\min} & \  \frac{(M-1)N_0\eta T P_s R_n(N_1,E_1)}{E_{2,n}R_n(N_1,E_1)+N_0M}+E_{2,n}.
\end{aligned}
\end{equation}
Problem \eqref{P:solveE2} is convex, whose optimal solution can be expressed in closed-form as
\begin{align}
E_{2,n}^\star(N_1,E_1)= & \left[\sqrt{\eta T P_s(M-1)N_0} -\frac{N_0M}{R_n(N_1,E_1)}\right]^+,\notag \\
 & n=1,\cdots, N_2.\label{eq:E2star}
\end{align}
The solution in \eqref{eq:E2star} implies that non-zero training energy should be allocated for the $n$th strongest sub-band in phase \rom{2} only if its expected channel power $R_n(N_1,E_1)$ is sufficiently large. Furthermore, \eqref{eq:E2star} has a nice water-filling interpretation, with fixed water level $\sqrt{\eta T P_s(M-1)N_0}$ and different base levels that depend on the expected channel power gain $R_n(N_1,E_1)$. For sub-bands with larger $R_n(N_1,E_1)$, more training energy should be allocated to achieve better MISO channel estimation, as expected.

The corresponding optimal value of \eqref{P:solveE2} is
\begin{equation}\label{eq:vnstar}
\small
\begin{aligned}
\hspace{-2ex} v_n^\star(N_1,E_1)=\begin{cases}
\frac{M-1}{M} \eta T P_s R_n(N_1,E_1)  & \hspace{-6ex} \text{ if } R_n(N_1,E_1)\leq \alpha,\\
2\sqrt{(M-1)N_0 \eta TP_s}-\frac{N_0 M}{R_n(N_1,E_1)}, & \text{ otherwise},
\end{cases}
\end{aligned}
\end{equation}
where $\alpha \triangleq \sqrt{N_0}M / \sqrt{\eta T P_s (M-1)}$.
By substituting the above optimal value of \eqref{P:solveE2} into \eqref{eq:Qnet}, the objective function of (P1) in terms of $N_1$ and $E_1$ only can be expressed as
\begin{align}\label{eq:QnetN1E1}
\Qnet(N_1,E_1) = \sum_{n=1}^{N_2} \Big[ \eta T P_s R_n(N_1, E_1) - v_n^\star(N_1,E_1)\Big]-N_1E_1.
\end{align}
As a result, (P1) reduces to
\begin{equation}\label{P:E1N1}
\begin{aligned}
\underset{E_1, N_1}{\max} \quad & \Qnet (N_1,E_1)  \\
\text{subject to} \quad &  N_2\leq N_1 \leq N, \ E_1\geq 0. 
\end{aligned}
\end{equation}

To find the optimal solution to problem \eqref{P:E1N1}, we first obtain the optimal training energy $E_1$ with $N_1$ fixed by solving the following problem,
\begin{equation}\label{P:E1}
\begin{aligned}
\underset{E_1\geq 0}{\max} \quad & \Qnet (N_1,E_1).
\end{aligned}
\end{equation}
Denote by $E_1^\star(N_1)$ and $\Qnet^\star(N_1)$ the optimal solution and optimal value of problem \eqref{P:E1}, respectively. The original problem (P1)  then reduces to determining the optimal number of training sub-bands $N_1$  as
\begin{align}
N_1^{\star}=\arg\underset{N_2\leq N_1 \leq N}{\max} \Qnet^{\star}(N_1),\label{eq:N1star}
\end{align}
which can be easily solved by one-dimensional search.
Therefore, the remaining task for solving (P1) is to find the optimal solution to problem \eqref{P:E1}. To achieve  this end, it is first noted  from \eqref{eq:Ehnsq} that for any fixed $N_1$, as the training energy $E_1$ varies from $0$ to $\infty$, $R_n(N_1,E_1)$ monotonically increases from $\beta M$ to $\beta G_n(N_1,M)$, i.e.,
\begin{align}
\beta M \leq R_n(N_1,E_1) \leq \beta G_n(N_1,M), \ \forall E_1\geq 0, \ n=1,\cdots, N_2. \label{eq:Rhbound}
\end{align}
Further define
\begin{align}
\Gamma\triangleq \frac{\eta T P_s \beta^2}{N_0} \label{eq:Gamma}
\end{align}
as the two-way effective signal-to-noise ratio (ESNR), where the term $\beta^2$ captures the effect of two-way signal attenuation due to both the reverse-link training and forward-link energy transmission. Based on the property of $G_n(N_1,M)$ given in \eqref{eq:Gorder}, problem \eqref{P:E1} can then be solved by separately considering the following three cases according to the intervals in which $\alpha$ lies, as illustrated in Fig.~\ref{F:ThreeCases}.

\begin{figure}
\centering
\includegraphics[scale=0.9]{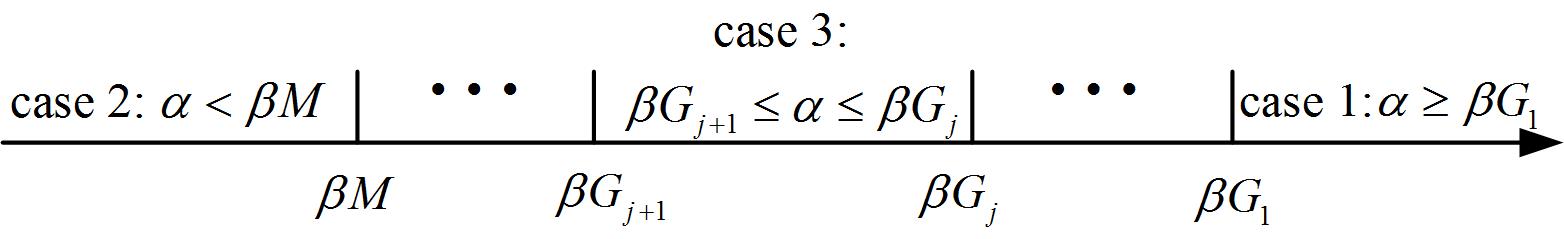}
\caption{Three different cases for solving problem \eqref{P:E1}.}\label{F:ThreeCases}
\end{figure}

{\it Case 1: $\alpha \geq \beta G_1 (N_1,M)$}, or $\Gamma \leq \frac{M^2}{(M-1)G_1^2(N_1,M)}$, which corresponds to the low-ESNR regime or the so-called harsh environment for WET by taking into account all the relevant factors including the channel conditions, the power limits at the ET, and the achievable energy-beamforming/frequency-diversity gains.  
  In this case, it follows from \eqref{eq:Rhbound} that $R_n(N_1,E_1)\leq \alpha$ and $E_{2,n}^\star(N_1,E_1)=0$, $\forall n=1,\cdots, N_2$ and $E_1\geq 0$. In other words, phase-\rom{2} training should not be performed due to the limited energy beamforming gains. Therefore, $\Qnet(N_1,E_1)$ in \eqref{eq:QnetN1E1} reduces to
\begin{align}\label{eq:QnetCase1}
\Qnet(N_1,E_1)=\eta T P_s \sum_{n=1}^{N_2} \frac{R_n(N_1,E_1)}{M}-E_1N_1, \ \forall E_1 \geq 0.
 \end{align}
By substituting $R_n(N_1,E_1)$ with  \eqref{eq:Ehnsq}, problem \eqref{P:E1} reduces to
\begin{equation}\label{P:E1case1}
\begin{aligned}
\underset{E_1\geq 0}{\max}\  \eta T P_s \beta  \sum_{n=1}^{N_2} \frac{\beta E_1 G_n(N_1,M)/M+ N_0}{\beta E_1 + N_0}-E_1N_1,
\end{aligned}
\end{equation}
which is a convex optimization problem with the optimal solution given by
\begin{align}
E_1^{\star}(N_1)=\sqrt{\eta T P_s N_0} \left[\sqrt{\frac{\sum_{n=1}^{N_2}\left(\frac{G_n(N_1,M)}{M}-1 \right)}{N_1}} -\frac{1}{\sqrt{\Gamma}}\right]^+.\label{eq:E1N1Case1}
\end{align}
The corresponding optimal value of \eqref{P:E1case1} can be expressed as
\begin{equation}\label{eq:QnetStarCase1}
\begin{aligned}
\Qnet^\star(N_1)=\begin{cases}
\eta T P_s \beta N_2, & \text{ if } \sum_{n=1}^{N_2} \left(\frac{G_n(N_1,M)}{M}-1\right)<\frac{N_1}{\Gamma} \\
\eta T P_s \beta\left( N_2 + \left(\sqrt{\sum_{n=1}^{N_2} \big(\frac{G_n(N_1,M)}{M}-1 \big)}-\sqrt{\frac{N_1}{\Gamma}}\right)^2\right), & \text {otherwise}.
\end{cases}
\end{aligned}
\end{equation}

It is observed from \eqref{eq:E1N1Case1} that for any fixed $N_1$, non-zero training energy should be allocated in phase \rom{1} only if the frequency-diversity gain normalized by $N_1$ is sufficiently large. In this case, the {\it net} harvested energy is given by a summation of that achieved in the case without CSI at the ET, $\eta T P_s \beta N_2$, and an additional gain due to the exploitation of frequency diversity, as shown in \eqref{eq:QnetStarCase1}.



{\it Case 2: $\alpha < \beta M$}, or $\Gamma > 1/(M-1)$, which corresponds to the high-ESNR regime or a favorable WET environment. In this case, it follows from \eqref{eq:Rhbound} that $R_n(N_1,E_1)  > \alpha$ and $E_{2,n}^\star(N_1,E_1)>0$, $\forall n=1,\cdots, N_2$ and $E_1\geq 0$. In other words, regardless of the training energy $E_1$ in phase \rom{1}, non-zero training energy should be allocated to each of the $N_2$ selected sub-bands in phase \rom{2} due to the potentially large beamforming gains.  As a result, it follows from \eqref{eq:vnstar} that $\Qnet(N_1,E_1)$ in \eqref{eq:QnetN1E1} can be expressed as
\begin{align}
\Qnet(N_1,E_1)=&\sum_{n=1}^{N_2} \Big[\eta T P_s R_n(N_1,E_1)- 2 \sqrt{(M-1)N_0 \eta T P_s} \notag \\
 &+ \frac{N_0M}{R_n(N_1,E_1)} \Big]-N_1E_1. \label{eq:QnetN1E1Case2}
\end{align}
By substituting $R_n(N_1,E_1)$ with \eqref{eq:Ehnsq} and after discarding constant terms, problem \eqref{P:E1} can be reformulated as
\begin{equation}\label{P:E1case2}
\begin{aligned}
\underset{E_1\geq 0}{\max}\  \eta & T P_s \beta  \sum_{n=1}^{N_2} \frac{\beta E_1 G_n(N_1,M)+ N_0 M}{\beta E_1
 + N_0} \\
 & +\frac{N_0M}{\beta} \sum_{n=1}^{N_2} \frac{\beta E_1 + N_0}{\beta E_1 G_n(N_1,M)+ N_0 M}-E_1N_1.
\end{aligned}
\end{equation}

Problem \eqref{P:E1case2} belongs to the class of optimizing the sum of linear fractional functions, which is difficult to be optimally solved in general. However, with some manipulations, problem \eqref{P:E1case2} can be simplified into
\begin{align}\label{P:caseIISimp}
\underset{E_1\geq 0}{\min}\  f(E_1)\triangleq \frac{b_0}{E_1+c_0}+E_1N_1-\sum_{n=1}^{N_2} \frac{b_n}{E_1+ c_n} ,
\end{align}
where
\begin{equation}\label{eq:coeffs}
\begin{aligned}
b_0&=\eta T P_s N_0 \sum_{n=1}^{N_2} \left(G_n(N_1,M)-M \right), \ c_0=\frac{N_0}{\beta},\\
b_n&=\frac{N_0^2 M \left( 1- \frac{M}{G_n(N_1,M)}\right)}{\beta^2 G_n(N_1,M)}, \ c_n=\frac{N_0 M}{\beta G_n(N_1,M)}, \ \forall n.
\end{aligned}
\end{equation}

Problem \eqref{P:caseIISimp} is still non-convex due to the non-convex objective function. However, by examining its Karush-Kuhn-Tucker (KKT) necessary conditions for global optimality (not necessarily sufficient due to non-convexity in general), the optimal solution to problem \eqref{P:caseIISimp} must belong to the set
\begin{align}
\mathcal{E}_1\triangleq \{0\}\cup \left \{E_1>0 \Big | \frac{\partial f(E_1)}{\partial E_1}=0 \right \}. \notag
\end{align}
By taking the derivative of the function $f(E_1)$ in \eqref{P:caseIISimp} and after some algebraic manipulations, it can be shown that finding the stationary points satisfying $\frac{\partial f(E_1)}{\partial E_1}=0$ is equivalent to determining the roots of the following polynomial equation,
\begin{equation*}
\begin{aligned}
N_1(E_1+c_0)^2  \prod_{i=1}^{N_2} & (E_1+c_i)^2  -b_0 \prod_{i=1}^{N_2} (E_1+c_i)^2\\
& + (E_1+c_0)^2\sum_{n=1}^{N_2} b_n \prod_{i=1,i\neq n}^{N_2} (E_1+c_i)^2=0,
\end{aligned}
\end{equation*}
which has order $2N_2+2$. The cardinality of $\mathcal{E}_1$ is thus bounded by $|\mathcal{E}_1|\leq 2N_2+3$. Therefore, the optimal solution to problem \eqref{P:caseIISimp}  can be readily obtained by comparing at most $2N_2+3$ candidate solutions as
\begin{align}
E_1^\star(N_1)=\arg \underset{E_1 \in \mathcal{E}_1}{\min} f(E_1).\label{eq:E1starCase2}
\end{align}
With $E_1^\star(N_1)$ determined, the corresponding optimal value of problem \eqref{P:E1} for the case $\alpha < \beta M$ can be obtained by direct substitution based on \eqref{eq:QnetN1E1Case2}.

{\it Case 3: $\beta G_{j+1}(N_1,M) \leq  \alpha \leq \beta G_j (N_1,M)$},  or $\frac{M^2}{(M-1)G_{j}^2(N_1,M)}\leq \Gamma \leq \frac{M^2}{(M-1)G_{j+1}^2(N_1,M)}$, for some $j=1,\cdots, N_2$. This corresponds to the medium-ESNR regime or a moderate WET environment. For convenience, we define $G_{N_2+1}(N_1,M)\triangleq \beta M$. Based on \eqref{eq:Gorder} and \eqref{eq:Rhbound}, we have $R_n(N_1,E_1)\leq \alpha$, $\forall n=j+1,\cdots, N_2$ and $E_1\geq 0$. 
On the other hand, for $n=1,\cdots, j$, $R_n(N_1,E_1)$ may be larger or smaller than $\alpha$, depending on the training energy $E_1$. Let $E_1$ be chosen such that $R_{k+1}(N_1,E_1)\leq  \alpha \leq R_k(N_1,E_1)$ for some $k=0,\cdots, j$. Then based on \eqref{eq:Ehnsq}, we must have $E_1\in \Big[E_1^{(k)}, E_1^{(k+1)}\Big]$, where $E_1^{(k)}\triangleq \frac{N_0(\alpha-\beta M)}{\beta (\beta G_k(N_1,M)-\alpha)}$. For convenience, we define $E_1^{(j+1)}=\infty$ and $E_1^{(0)}=0$. For $E_1$ belonging to the above interval, we have $R_n(N_1,E_1)\geq \alpha$, $\forall n=1,\cdots, k$, and $R_n(N_1,E_1)\leq \alpha$, $\forall n=k+1,\cdots, j$. Therefore, depending on the $j+1$ intervals in which $E_1$ lies, $\Qnet(N_1,E_1)$ in \eqref{eq:QnetN1E1} can be explicitly written as
\begin{align}\label{eq:QnetCase3}
 \Qnet(N_1,E_1)=f_k(E_1), \text{ for } E_1\in \big[E_1^{(k)}, E_1^{(k+1)}\big], k=0,\cdots, j,
\end{align}
where
\begin{equation}\label{eq:fk}
\small
\begin{aligned}
& f_k(E_1)\triangleq \sum_{n=1}^k\Big(\eta T P_s R_n(N_1, E_1)-2\sqrt{(M-1)N_0\eta T P_s} + \frac{N_0M}{R_n(N_1,E_1)}\Big) \\
&+ \sum_{n=k+1}^{N_2} \left(\eta T P_s R_n(N_1,E_1)-\frac{M-1}{M}\eta T P_s R_n(N_1,E_1) \right)-N_1E_1.
\end{aligned}
\end{equation}

As a result,  problem \eqref{P:E1} for Case 3 can be decomposed into $j+1$ parallel sub-problems by optimizing $E_1$ over each of the $j+1$ intervals, with the $k$th sub-problem formulated as
\begin{equation}\label{P:E1case3}
\begin{aligned}
\underset{E_1}{\max} & \  f_k(E_1)\\
\text{subject to} & \ E_1^{(k)}\leq E_1\leq E_1^{(k+1)}.
\end{aligned}
\end{equation}
Denote by $E_{1,k}^\star$ the optimal solution to the $k$th sub-problem \eqref{P:E1case3}. The solution $E_1^\star$ to problem \eqref{P:E1} is thus given by
\begin{align}
E_1^\star(N_1)=\arg \underset{E_{1,k}^\star, k=0,\cdots, j}{\max}f_k(E_{1,k}^\star),\label{eq:E1starCase3}
\end{align}
which can be readily determined. Thus, the remaining task is to solve problem \eqref{P:E1case3}. With some simple manipulations and after discarding irrelevant terms, \eqref{P:E1case3} can be reformulated as
\begin{align}\label{P:caseIIISimp}
\underset{E_1^{(k)}\leq E_1\leq E_1^{(k+1)}}{\min}\  \frac{d_0}{E_1+c_0}+E_1N_1-\sum_{n=1}^{k} \frac{b_n}{E_1+ c_n} ,
\end{align}
where $c_0$, $b_n$, and $c_n$, for $n=1,\cdots, k$, are defined in \eqref{eq:coeffs}, and
\begin{equation}
\small
\begin{aligned}
\hspace{-2ex} d_0&\triangleq\eta T P_s N_0 \Big(\sum_{n=1}^{k} (G_n(N_1,M)-M )+\sum_{n=k+1}^{N_2}\big(\frac{G_n(N_1,M)}{M}-1\big)\Big). \notag
\end{aligned}
\end{equation}
Note that problem \eqref{P:caseIIISimp} has identical structure as problem \eqref{P:caseIISimp} and thus can be optimally solved with the same technique as discussed above.

The algorithm for optimally solving the training optimization problem (P1) is summarized as follows.

\begin{algorithm}[H]
\caption{Optimal Solution to (P1)}
\label{Algo:P1}
\begin{algorithmic}[1]
\FOR{$N_1=N_2,\cdots, N$}
\IF{$\alpha \geq \beta G_1(N_1,M)$}
 \STATE Obtain $E_1^\star(N_1)$ and $\Qnet^\star(N_1)$ based on \eqref{eq:E1N1Case1} and \eqref{eq:QnetStarCase1}, respectively.
 \ELSIF {$\alpha < \beta M$}
 \STATE Obtain $E_1^\star(N_1)$ and the corresponding  $\Qnet^\star(N_1)$ based on \eqref{eq:E1starCase2} and \eqref{eq:QnetN1E1Case2}, respectively.
 \ELSE
 \STATE Determine $j\in \{1,\cdots, N_2\}$ such that $\beta G_{j+1}(N_1,M)\leq \alpha \leq \beta G_j(N_1,M)$
 \FOR{$k=0,\cdots, j$}
 \STATE Find $E_{1,k}^\star$ by solving \eqref{P:caseIIISimp} and obtain the corresponding $f_k(E_{1,k}^\star)$ based on \eqref{eq:fk}.
 \ENDFOR
 \STATE Obtain $E_1^\star(N_1)$ and the corresponding $\Qnet^\star(N_1)$ based on  \eqref{eq:E1starCase3} and \eqref{eq:QnetCase3}, respectively.
 \ENDIF
 \ENDFOR
 \STATE Obtain $N_1^\star$ and the corresponding $E_1^\star$ based on \eqref{eq:N1star}.
 \STATE Obtain $E _{2,n}^\star$, $n=1,\cdots, N_2$, based on \eqref{eq:E2star} with $N_1^\star$ and $E_1^\star$.
\end{algorithmic}
\end{algorithm}

Algorithm~\ref{Algo:P1} can be efficiently implemented since it has a polynomial complexity, as we show next. First, it is noted that for each iteration of the outer ``for'' loop, the worst-case scenario corresponds to the execution of the third case (lines 7-10). The interval searching in line 7 can be easily implemented with the bisection method of complexity $O(\log N_2)$, where $O(\cdot)$ is the standard ``big O'' notation. Line 9 involves solving polynomial equations of order $2N_2+2$, for which there exist various efficient algorithms of complexity no larger than $O((2N_2+2)^3)$ \cite{556}. Furthermore, since in the worst-case scenario, the outer (line 1) and inner (line 8) ``for'' loops will be executed by $N-N_2+1$ and $N_2+1$ times, respectively, the overall complexity of Algorithm~\ref{Algo:P1} can be expressed as $O((N-N_2+1)[\log N_2 + (N_2+1)(2N_2+2)^3])$, or $O((N-N_2)N_2^4)$ for large $N$ and $N_2$.

\section{Asymptotic Analysis}\label{sec:asympt}
In this section,   we provide an asymptotic analysis for the training optimization problem (P1) with either a large number of ET antennas $M$ (massive-array WET) or a large number of sub-bands $N$ (massive-band WET), from which we draw important insights.
\subsection{Massive-Array WET}
Massive MIMO has emerged as one of the key enabling technologies for 5G wireless communication systems due to its tremendous improvement on energy and spectrum efficiency by deploying a very large number of antennas (say, hundreds or even more) at the base stations \cite{374,497}.  For WET systems with a large number of antennas deployed at the ET, it is expected that the energy transfer efficiency can be similarly enhanced, provided that the channel is properly learned at the ET to achieve the enormous beamforming gain. To obtain the asymptotic analysis 
for (P1) with $M\gg 1$, we first present the following lemma.
\begin{lemma}\label{lemma:GnlargeM}
For the function $G_n(N_1,M)$ defined in Lemma~\ref{lemma:expV} in Appendix~\ref{A:usefulLemma}, the following asymptotic result holds:
\begin{align}
G_n(N_1,M){\rightarrow} M, \ \forall n, N_1, \text{as } M\rightarrow \infty. \label{eq:GnMassive}
\end{align}
\end{lemma}
\begin{IEEEproof}
  For the $N_1$ i.i.d. CSCG random vectors $\mathbf v_i\sim \mathcal{CN}(\mathbf 0, \sigma_v^2 \mathbf I_M)$  given in Lemma~\ref{lemma:expV},  the following asymptotic result holds as $M\rightarrow \infty$ due to the law of large numbers:
\begin{align}
\frac{1}{M} \|\mathbf v_i\|^2 \rightarrow \sigma_v^2,\  \forall i=1,\cdots, N_1. \notag
\end{align}
Therefore, we have $\xE[\|\mathbf v_{[n]}\|^2]\rightarrow \sigma_v^2 M$, $\forall n$, which thus leads to \eqref{eq:GnMassive} by definition.
\end{IEEEproof}

\begin{lemma}\label{lemma:largeM}
For MISO WET systems in frequency-selective  channels with $M$ transmit antennas and $N$ independent sub-bands, as $M\rightarrow \infty$, the optimized two-phase training scheme reduces to
\begin{align}
&E_1^{\text{large-}M} \rightarrow 0, \label{eq:E1largeM}\\
&E_{2,n}^{\text{large-}M}\rightarrow \sqrt{\eta T P_sN_0 M}, \ n=1,\cdots, N_2. \label{eq:E2largeM}
\end{align}
Furthermore, the resulting net harvested energy scales linearly with $M$ as
\begin{align}
\Qnet^{\text{large-}M} \rightarrow  \eta T N_2 P_s \beta M.\label{eq:QnetLargeM}
\end{align}
\end{lemma}
\begin{IEEEproof}
Please refer to Appendix~\ref{A:largeM}.
\end{IEEEproof}

Lemma~\ref{lemma:largeM} shows that due to the so-called ``channel hardening'' effect for large $M$, i.e., $\|\mathbf h_n\|^2 \rightarrow \beta M$, $\forall n=1,\cdots, N$, as $M\rightarrow \infty$, all the $N$ sub-bands have essentially identical power gain $\beta M$. Therefore, the frequency-diversity gain vanishes and hence phase-\rom{1} training is no longer required, as shown in \eqref{eq:E1largeM}. Moreover, it is observed from \eqref{eq:E2largeM} that as $M$ increases, the training energy in phase \rom{2} for each of the selected sub-band  scales with $M$ in a square root manner. Since the beamforming gain increases linearly with $M$, which dominates the training energy cost at large $M$, the \emph{net} harvested energy scales linearly with $M$, as shown in \eqref{eq:QnetLargeM}.

\begin{lemma}\label{lemma:idealLargeM}
For MISO WET systems in frequency-selective  channels with $M$ transmit antennas and $N$ independent sub-bands, as $M\rightarrow \infty$, the average harvested energy in the ideal case with perfect CSI at the ET scales with $M$ as
\begin{align}
\bar{Q}^{\text{ideal}} \rightarrow  \eta T N_2 P_s   \beta M.
\end{align}
\end{lemma}
\begin{IEEEproof}
With perfect CSI at the ET, the maximum harvested energy for each channel realization has been obtained in \eqref{eq:Qmax}. The average value taken over all channel realizations can then be  obtained as
\begin{align}
\bar{Q}^{\text{ideal}}&=\xE\left[ Q^{\text{ideal}}\right]=\eta T P_s \sum_{n=1}^{N_2} \xE\left[ \|\mathbf h_{[n]}\|^2\right]\notag \\
&=\eta T P_s \sum_{n=1}^{N_2} \beta G_n(N,M) \label{eq:Qideal1}\\
&\rightarrow \eta T N_2 P_s  \beta M, \label{eq:Qideal2}
\end{align}
where \eqref{eq:Qideal1} follows from the definition of the function $G_n(N_1,M)$ given in Lemma~\ref{lemma:expV} in Appendix~\ref{A:usefulLemma}, and \eqref{eq:Qideal2} is true due to Lemma~\ref{lemma:GnlargeM}.
\end{IEEEproof}


It immediately follows from Lemma~\ref{lemma:largeM} and Lemma~\ref{lemma:idealLargeM} that for MISO WET systems in frequency-selective channels with $M$ transmit antennas and $N$ independent sub-bands, the proposed two-phase training scheme achieves the optimal asymptotic scaling with $M$.

\subsection{Massive-Band WET}
In this subsection, we investigate the WET system with an asymptotically large bandwidth $B$, or equivalently, with an unlimited number of sub-bands $N$. 
We first study the asymptotic scaling of the average harvested energy at the ER with $N$ under the ideal case with perfect CSI at the ET.

\begin{lemma}\label{lemma:largeNPerfect}
For MISO WET systems in frequency-selective channels with $M$ transmit antennas and $N$ independent sub-bands, as $N\rightarrow \infty$, the average harvested energy $\bar{Q}^{\text{ideal}}(N)$ in the ideal case with perfect CSI at the ET scales logarithmally with $N$, i.e., $\bar{Q}^{\text{ideal}}(N)=\Theta(\ln N)$.
\end{lemma}
\begin{IEEEproof}
Please refer to Appendix~\ref{A:largeNPerfect}.
\end{IEEEproof}


\begin{lemma}\label{lemma:largeNProposed}
For MISO WET systems in frequency-selective channels with $M$ transmit antennas and $N$ independent sub-bands, as $N\rightarrow \infty$, the net harvested energy $\bar{Q}^{\text{large-}N}$ with the proposed two-phase training scheme is upper-bounded by
\begin{equation}
\small
\begin{aligned}
\bar{Q}^{\text{large-}N} \leq \eta T N_2 P_2 \beta M \bigg( 1 +\Big(\sqrt{W(\Gamma N_2 M)}-\frac{1}{\sqrt{W(\Gamma N_2 M)}} \Big)^2\bigg), \label{eq:QnetUBProposed}
\end{aligned}
\end{equation}
where $\Gamma$ is the two-way ESNR defined in \eqref{eq:Gamma}, and $W(z)$ is the Lambert W function defined by $z=W(z)e^{W(z)}$.
\end{lemma}

\begin{IEEEproof}
Please refer to Appendix~\ref{A:largeN}.
\end{IEEEproof}

Lemma~\ref{lemma:largeNProposed} shows that, different from the asymptotic linear scaling with the number of transmit antennas $M$ as in Lemma~\ref{lemma:largeM}, the \emph{net} harvested energy at the ER does not asymptotically increase with  $N$ as $N$ becomes large; instead, it is upper-bounded by a constant value that depends on the two-way ESNR, $\Gamma$. This result is unsurprising, since while the energy cost $(E_1N_1)$ for channel training in phase \rom{1} increases linearly with the number of trained sub-bands,  the frequency-diversity gain scales only in a logarithmic manner with $N_1$, as shown in Lemma~\ref{lemma:largeNPerfect}. Consequently, beyond some large value of $N$, the benefit of channel training for exploiting the frequency-diversity gain cannot compensate the training energy cost; thus, further bandwidth expansion does not provide any performance improvement in terms of the \emph{net} harvested energy at the ER.

\begin{figure}
\centering
\includegraphics[scale=0.25]{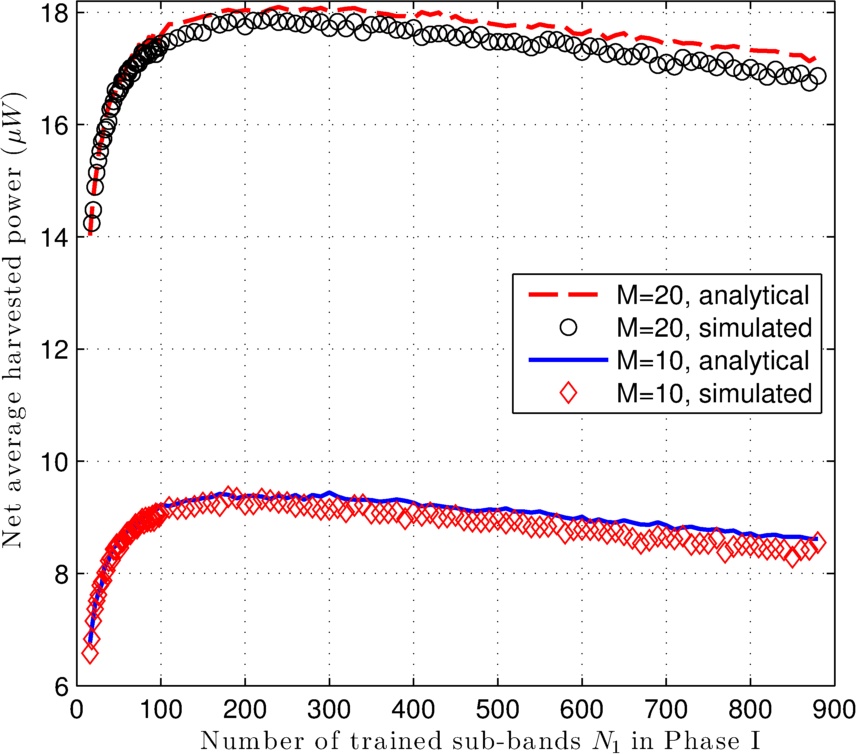} 
\caption{Net average harvested power versus the number of trained sub-bands $N_1$ for $M=10$ and $M=20$.
}\label{F:QnetVsN1}
\end{figure}

\section{Numerical Results}\label{sec:simulation}
In this section, numerical examples are provided to corroborate our study.  To model the frequency-selective channel,  we assume a multi-path power delay profile with exponential distribution  $A(\tau)=\frac{1}{\sigma_{\text{rms}}} e^{-\tau/\sigma_{\text{rms}}}$, $\tau\geq 0$, where $\sigma_{\text{rms}}$  denotes the root-mean-square (rms) delay spread. We set $\sigma_{\text{rms}}=1\mu$s so that the $50\%$ channel coherence bandwidth \cite{523}, i.e., the frequency separation for which the amplitude correlation is $0.5$,  is $B_c=\frac{1}{2\pi \sigma_{\text{rms}}}\approx 160$ kHz. We assume that the system is operated over the 902-928MHz ISM  band, which is equally divided into $N=866$ sub-bands each with bandwidth $B_s=30$kHz.  To comply with the FCC (Federal Communications Commission) regulations over this band \cite{549}, (Part 15.247, paragraph (e): the power spectral density from the intentional radiator ``shall not be greater than $8$dBm in any $3$kHz band''), we set the transmission power for each sub-band as $P_s=60$mW. Furthermore, according to the FCC regulations Part 15.247, paragraph (b) (3), the maximum output power over the 902-928MHz band shall not exceed 1watt, i.e., $P_t\leq 1$watt. Therefore, at each instance, the maximum number of sub-bands that can be activated for energy transmission is $N_2=\lfloor P_t/P_s \rfloor=16$. We  assume that the power spectrum density of the training noise received at the ET is $N_0=-160$dBm/Hz. The average power attenuation between the ET and the ER is assumed to be $60$dB, i.e., $\beta=10^{-6}$, which corresponds to an operating distance about $30$ meters for carrier frequency around $900$MHz. Furthermore, we assume that the  energy conversion efficiency at the ER is $\eta=0.8$.

In Fig.~\ref{F:QnetVsN1}, by varying the number of sub-bands $N_1$ that are trained in phase \rom{1}, the net average harvested power achieved by the optimized two-phase training scheme is plotted for $M=10$ and $M=20$, where the average is taken over $10^4$ random channel realizations. The channel block length is set as $T=0.05$ms.
 The analytical result obtained in Section~\ref{sec:optimalTraining}, i.e.,  $\Qnet^{\star}(N_1)/T$ with $\Qnet^{\star}(N_1)$ denoting the optimal value of problem \eqref{P:E1},  is also shown in Fig.~\ref{F:QnetVsN1}. It is observed that the simulation and analytical results match quite well with each other, which thus validates our theoretical studies. Furthermore, Fig.~\ref{F:QnetVsN1} shows that there is an optimal number of sub-bands to be trained for maximizing the net harvested energy, due to the non-trivial trade-off between achieving more frequency-diversity gain (with lager $N_1$) and reducing the training energy ($E_1N_1$ in phase \rom{1}).

 \begin{figure}
\centering
\includegraphics[scale=0.25]{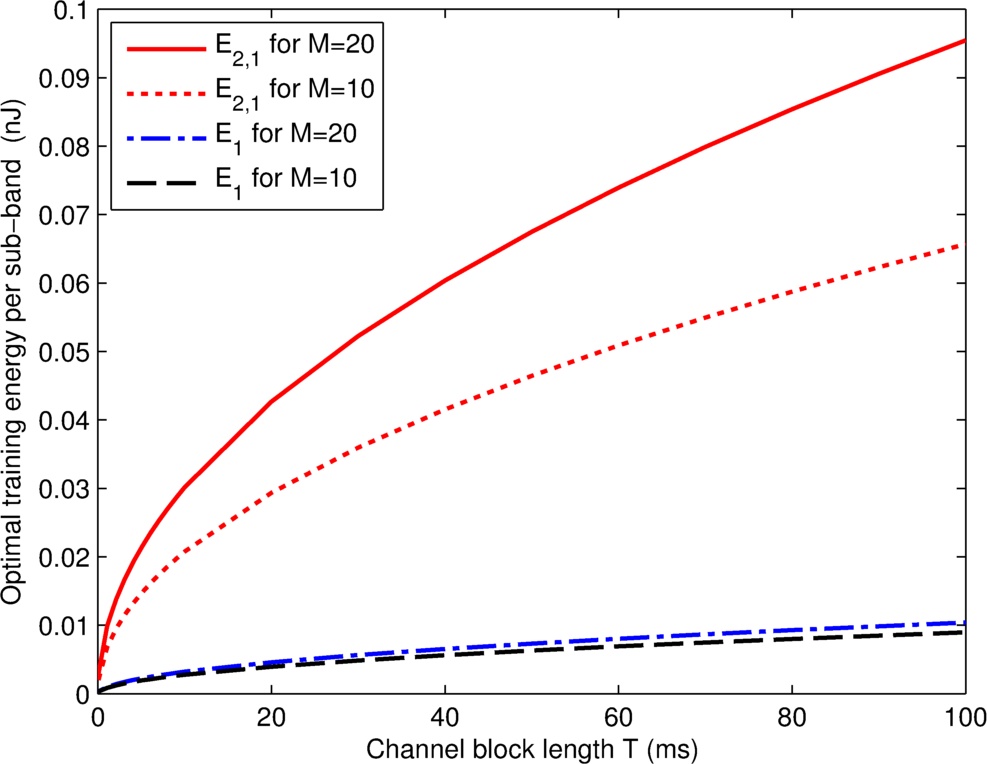}
\caption{Optimal training energy $E_1$ and $E_{2,1}$ versus block length $T$ for $M=10$ and $M=20$.}\label{F:E1E2VsT}
\end{figure}

 In Fig.~\ref{F:E1E2VsT}, the optimal training energy per sub-band $E_1$ in phase \rom{1} and that for the strongest sub-band $E_{2,1}$ in phase \rom{2} are plotted against the channel block length $T$, with $T$ ranging from $0.1$ to $100$ms. It is observed that $E_1$ and $E_{2,1}$ both increase with $T$, which is expected since as the block length increases, more training energy is affordable  due to the increased energy harvested during each block. It is also observed from Fig.~\ref{F:E1E2VsT} that more training energy should be allocated for system with $M=20$ than that with $M=10$, since energy beamforming is more effective and hence training is more beneficial when more antennas are available at the ET.
  Furthermore, for both cases of $M=10$ and $M=20$, $E_{2,1}$ is significantly larger than $E_1$, since in phase \rom{2}, only the $N_2$ selected sub-bands need to be further trained, whereas the training energy in phase \rom{1} needs to be distributed over $N_1\gg N_2$ sub-bands to exploit the frequency diversity.


In Fig.~\ref{F:PnetVsTM20N216}, the net average harvested power based on the proposed two-phase training scheme is plotted against block length $T$, with $T$ ranging from $0.1$ to $20$ms. The number of antennas at the ET is $M=20$. The following benchmark schemes are also considered for comparison.
 \begin{enumerate}
\item {\it Perfect CSI:} 
In this case, ideal energy beamforming as given in \eqref{eq:optxn} can be performed and the average harvested energy is given in \eqref{eq:Qideal1}. This corresponds to the performance upper bound since the frequency-diversity and energy-beamforming gains are maximally  attained without any energy cost for CSI acquisition.  
\item {\it No CSI:} In this case, $N_2$ out of the $N$ total available  sub-bands are randomly  selected and isotropic energy transmission is performed over each of the selected sub-band.  The resulting average harvested energy at the ER can be expressed as
    $\bar{Q}_{\text{noCSI}}=\eta T P_s \beta N_2.$ 
    Neither frequency-diversity nor beamforming gain is achieved in this case.
\item {\it Phase-\rom{1} training only:} This corresponds to the special case of the proposed two-phase training scheme with $E_{2,n}=0$, $\forall n$. Thus, only frequency-diversity gain is exploited.
\item {\it Phase-\rom{2} training only:} This corresponds to the special case of the two-phase training scheme with $E_{1}=0$. Thus, only energy beamforming gain is exploited.
\item {\it Brute-force training:} This refers to the single-phase training scheme where the exact MISO channels at all the $N$ available sub-bands are estimated. Both frequency-diversity and energy-beamforming gains are exploited by this scheme, but it generally requires the most training energy among all the training schemes.
\end{enumerate}

\begin{figure}
\centering
\includegraphics[scale=0.25]{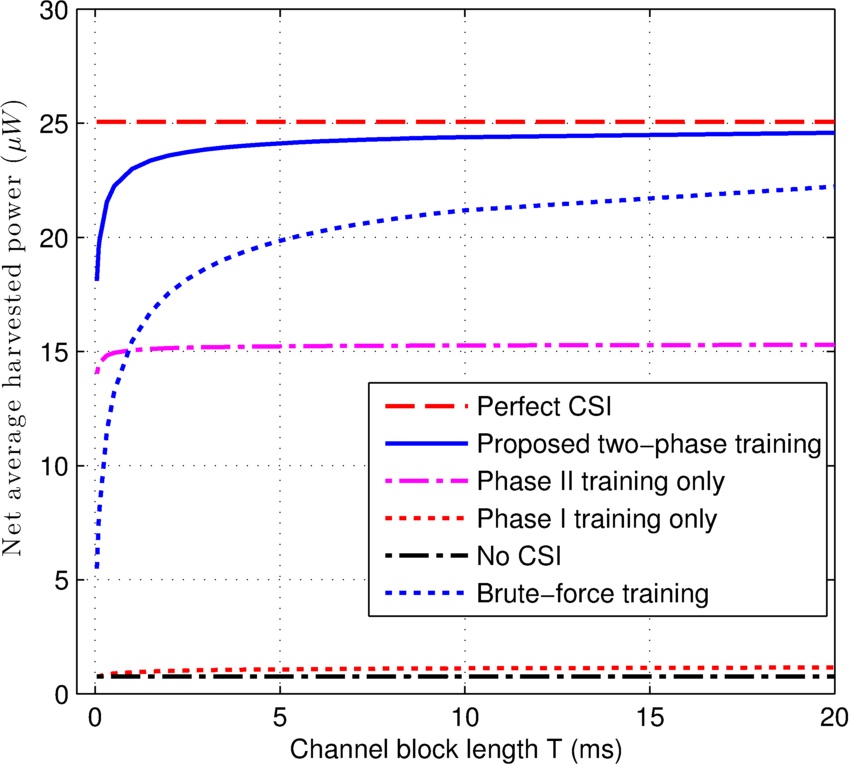}
\caption{Net average harvested power versus block length $T$ for $M=20$.}\label{F:PnetVsTM20N216}
\end{figure}


  It is observed from Fig.~\ref{F:PnetVsTM20N216} that the proposed two-phase training  scheme significantly outperforms the other four benchmark schemes. As the block length $T$ increases, both the two-phase training and the brute-force training schemes approach the performance upper bound with perfect CSI, which is expected since the  energy cost due to channel training becomes less significant with increasing $T$. It can also be inferred from the figure that for multi-antenna  WET in frequency-selective channels, exploiting either frequency-diversity (with phase-\rom{1} training only) or energy beamforming  (with phase-\rom{2} training only) gain alone is far from optimal in general; instead, a good balance between these  two gains as achieved in the proposed two-phase training optimization  is needed.

 \begin{figure}
\centering
\includegraphics[scale=0.25]{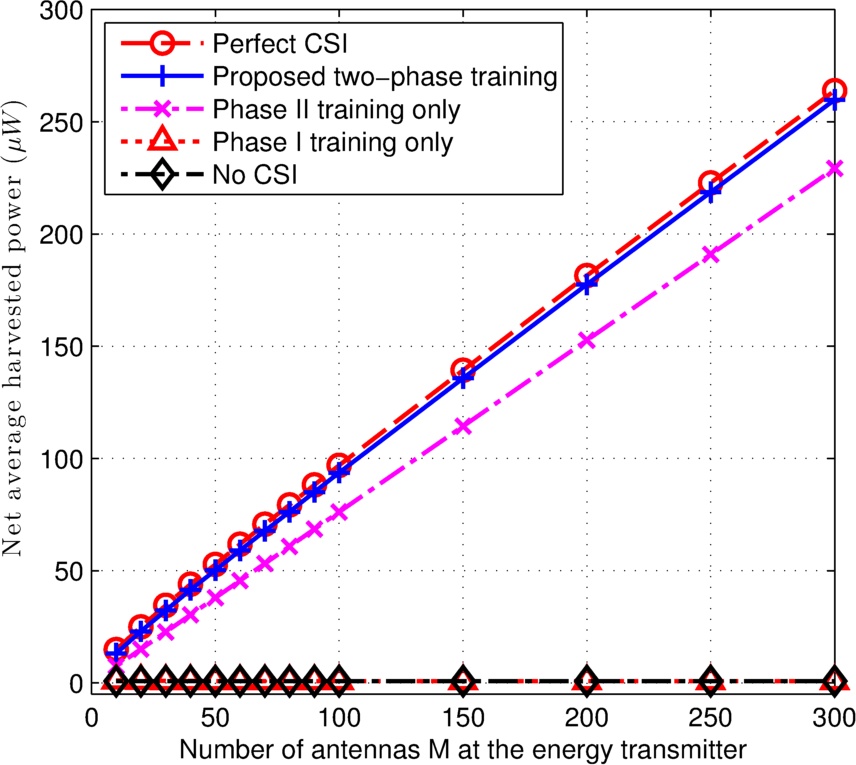}
\caption{Net average harvested power versus number of antennas $M$ at the ET. }\label{F:QnetVsM}
\end{figure}

In Fig.~\ref{F:QnetVsM}, to investigate the massive-array case with large $M$, the net average harvested power is plotted against the number of antennas $M$ at the ET for $T=1$ms. It is observed that both the schemes of  no-CSI and  phase-\rom{1} training only  result in poor performance.  On the other hand, for the two training schemes that exploit energy beamforming gain (two-phase training and phase-\rom{2} training only), the \emph{net} average harvested power scales linearly with $M$, which is in accordance to Lemma~\ref{lemma:largeM}. 
It is also observed from Fig.~\ref{F:QnetVsM} that, under this setup, the optimized two-phase training still provides considerable performance gain over the phase-\rom{2} training only, although  the performance gap becomes inconsequential for extremely large $M$ due to the linear scaling of both schemes, as expected from Lemma~\ref{lemma:largeM}.


In Fig.~\ref{F:QnetVsNSISO}, to investigate the massive-band case with large $N$, the net average harvested power is plotted against the total number of available sub-bands $N$ for the case of SISO WET with $N_2=1$ and $M=1$ (thus, phase-\rom{2} training is no longer needed and the proposed two-phase training is the same as phase-\rom{1} training only), where $N$ varies from $20$ to $2000$ (by assuming that some other spectrum in addition to the 902-928MHz ISM band is available). The channel block length is set as $T=50$ms. It is observed that in the ideal case with perfect CSI at the ET, the average harvested power increases logarithmically with $N$, which is in accordance to Lemma~\ref{lemma:largeNPerfect}. In contrast, for the  proposed training scheme that takes into account the energy cost of channel training, the net harvested power saturates as $N$ becomes large, as expected from Lemma~\ref{lemma:largeNProposed}. Furthermore, Fig.~\ref{F:QnetVsNSISO} shows that in this case, the asymptotic result given in \eqref{eq:QnetUBProposed} provides a satisfactory performance upper bound for the proposed training scheme with large $N$. 


 \begin{figure}
\centering
\includegraphics[scale=0.25]{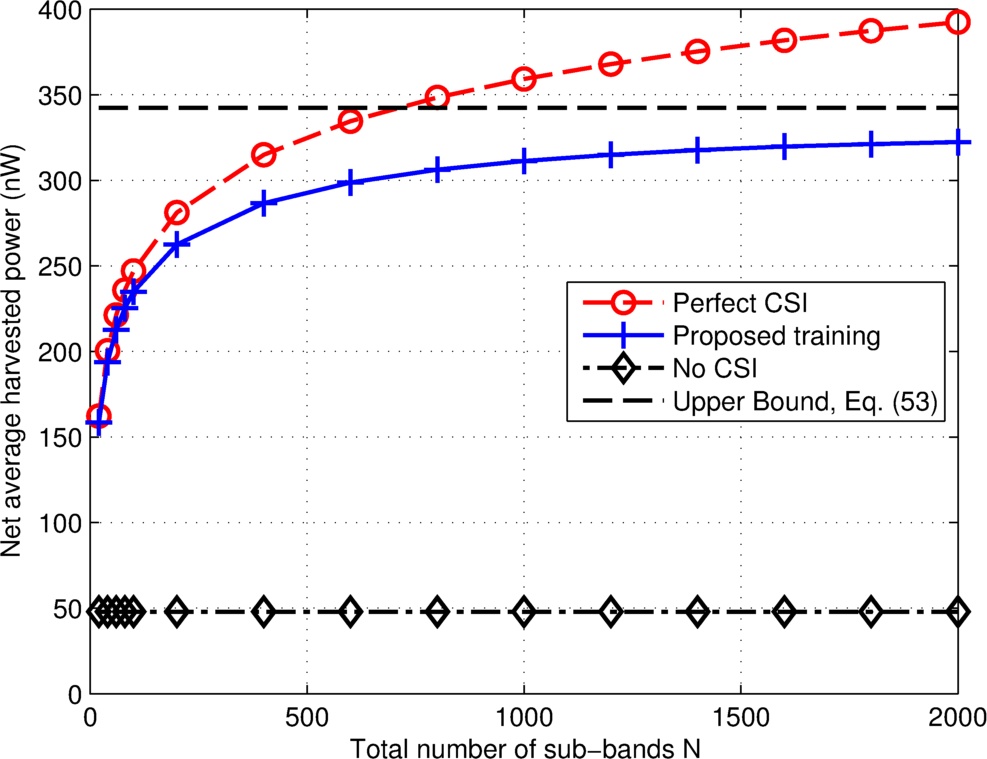}
\caption{Net average harvested power versus total number of sub-bands $N$ for SISO WET.}\label{F:QnetVsNSISO}
\end{figure}

 \section{Conclusions and Future Work}\label{sec:conclusion}
This paper studies the energy-constrained training design for MISO  WET systems in frequency-selective channels. By exploiting channel reciprocity, a two-phase training scheme is proposed. A closed-form expression is derived for the net harvested energy  at the ER, based on which a training design optimization problem is formulated and optimally solved. Asymptotic analysis is also presented for systems with a large number of transmit antennas or independent frequency sub-bands to gain useful insights. Numerical results are provided to validate our analysis and  demonstrate the effectiveness of the proposed scheme which optimally balances the achievable frequency-diversity and energy-beamforming gains in multi-antenna WET with limited-energy training.

There are a number of research directions along which the developed results in this paper can be further investigated, as briefly discussed in the following.
\begin{itemize}
\item {\it MIMO WET:} The optimal training design for the general MIMO wide-band WET systems remains an open problem. MIMO WET for flat-fading channels have been studied in \cite{528}. However, the extension to the more complex frequency-selective channels is no-trivial and deserves further investigation.
\item  {\it Correlated Channel:}  The analysis in this paper is based on the assumption of independent channels in both spatial and frequency domains, which practically holds with sufficient sub-band partition and antenna spacing. For the general setup with correlated channels, both the training optimization and performance analysis become more challenging; however,  spatial  and/or frequency channel correlations can also be exploited to reduce the training overhead and may enhance the overall energy transfer efficiency, which needs further investigation.
\item \emph{Multi-User Setup:} More  research endeavor is needed to find the optimal training design for the general multi-user setups with near-far located ERs. In this case, the so-called ``double near-far'' problem reported in \cite{515,528} needs to be addressed,  where a far-away ER from the ET suffers from higher propagation loss than a close  ER for both reverse-link channel training and forward-link energy transmission.
\item \emph{Energy Outage:} In this paper, the net average harvested power at the ER is maximized by optimizing the training design. In certain scenarios such as the delay-sensitive  applications, a more appropriate design criterion  may be minimizing the energy outage probability of ERs, for which the optimal training and energy transmission design still remain open.
\end{itemize}


\appendices
\section{Order Statistic}\label{A:orderStatistic}
This appendix gives a  brief introduction on order statistic, which is used for the proof of Lemma~\ref{lemma:exphnSq}. A comprehensive treatment on this topic can be found in e.g. \cite{546}.
\begin{definition}
Let $V_1,\cdots, V_{N_1}$ be $N_1$ real-valued random variables. The order statistics $V_{[1]},\cdots, V_{[N_1]}$ are random variables defined by sorting the values of $V_1,\cdots, V_{N_1}$ in non-increasing order, i.e., $V_{[1]}\geq V_{[2]}\geq \cdots \geq V_{[N_1]}$.
\end{definition}

\begin{lemma}
Let $V_1,\cdots, V_{N_1}$ be $N_1$ i.i.d. real-valued random variables with cumulative distribution function (CDF) $F(v)$. The CDF of the $n$th order statistic $V_{[n]}$, denoted as $F_{[n]}(v)$, can be expressed as
\begin{align}
F_{[n]}(v)=\sum_{k=0}^{n-1} \binom{N_1}{k} F(v)^{N_1-k}\left(1-F(v)\right)^k, \ n=1,\cdots, N_1. \label{eq:Fvn}
\end{align}
\end{lemma}

\section{A Useful Lemma}\label{A:usefulLemma}
\begin{lemma}\label{lemma:expV}
Let $\mathbf v_1,\cdots, \mathbf v_{N_1}\in \mathbb{C}^{M\times 1}$ be $N_1$ i.i.d. zero-mean CSCG random vectors distributed as $\mathbf v_n\sim \mathcal{CN}(\mathbf 0, \sigma^2_v \mathbf I_M)$, $\forall n$. Further denote by $[\cdot]$ the permutation such that $\|\mathbf v_{[1]}\|^2\geq \|\mathbf v_{[2]}\|^2\geq \cdots \geq \|\mathbf v_{[N_1]}\|^2$. Then we have
\begin{align}
\xE\left[\left\| \mathbf v_{[n]} \right\|^2 \right]=\sigma^2_v G_n(N_1,M), \ n=1,\cdots, N_1,\label{eq:expOrderstat}
\end{align}
where $G_n(N_1,M)$, $n=1,\cdots, N_1$, is a function of $N_1$ and $M$ given by
\begin{align}
G_1(N_1,M)&=\sum_{p=1}^{N_1} \binom{N_1}{p}(-1)^{p+1} c_p, \label{eq:G1} \\
G_{n+1}(N_1,M)&=G_{n}(N_1,M)-\Delta_n, \ n=1,\cdots, N_1-1,\label{eq:Gnplus1}
\end{align}
with
\begin{align}
c_p  =& \sum_{k_0+\cdots k_{M-1}=p}\binom{p}{k_0,\cdots,k_{M-1}} \left( \prod_{m=0}^{M-1}  \frac{1}{\left(m!\right)^{k_m}}\right) \notag \\
& \hspace{3ex} \times \left(\sum_{m=0}^{M-1}mk_m \right)! \frac{1}{p^{1+\sum_{m=0}^{M-1}mk_m}}, \label{eq:an} \\
\Delta_n = & \binom{N_1}{n} \sum_{l=0}^n \sum_{p=0}^{N_1-l} \binom{n}{l}\binom{N_1-l}{p}(-1)^{n-l+p}c_p.
\end{align}
Note that in \eqref{eq:an}, the summation is taken over all sequences of non-negative integer indices $k_0$ to $k_{M-1}$ with the sum equal to $p$, and the coefficients $\dbinom{p}{k_0,\cdots,k_{M-1}}$ are known as multinomial coefficients, which can be computed as
\begin{align}
\binom{p}{k_0,\cdots,k_{M-1}}=\frac{p!}{k_0!\cdots k_{M-1}!}. \notag
\end{align}
Furthermore, for $n=1$ and $M=1$, $G_n(N_1,M)$ can be expressed as the $N_1$th partial sum of the harmonic series, i.e.,
\begin{align}
G_1(N_1,1)=\sum_{i=1}^{N_1} \frac{1}{i}.\label{eq:G1M1}
\end{align}
\end{lemma}

\begin{proof}
Define the random variables $V_n\triangleq\|\mathbf v_n\|^2$, $n=1,\cdots, N_1$. It then follows that $V_1,\cdots, V_{N_1}$ are i.i.d. Erlang distributed with shape parameter $M$ and rate $\lambda=1/\sigma^2_v$, whose CDF is given by
\begin{align}
F(v)=1-\sum_{m=0}^{M-1} \frac{1}{m!} e^{-\lambda v}(\lambda v)^m, \ v\geq 0. \label{eq:FvErlang}
\end{align}
The CDF $F_{[n]}(v)$ of the order statistics $V_{[n]}\triangleq \|\mathbf v_{[n]}\|^2$ is then obtained by substituting $F(v)$ in \eqref{eq:FvErlang} into \eqref{eq:Fvn}. The expectations can then be expressed as
\begin{align}
\xE \left[ V_{[n]}\right]= \int_{0}^\infty \left(1-F_{[n]}(v)\right)dv, \ n =1,\cdots, N_1. \label{eq:Evn}
\end{align}
We first show \eqref{eq:expOrderstat} for $n=1$, i.e., the expectation of the maximum of $N_1$ i.i.d. Erlang random variables.  By substituting $n=1$ into \eqref{eq:Fvn}, we have $F_{[1]}(v)=F(v)^{N_1}$. With binomial expansion, we then have
\begin{align}
\xE\left[V_{[1]}\right]&=\int_{0}^\infty \left(1-\left(1-\sum_{m=0}^{M-1} \frac{1}{m!} e^{-\lambda v}(\lambda v)^m\right)^{N_1}\right)dv \notag \\
&=\sum_{p=1}^{N_1} \binom{N_1}{p} (-1)^{p+1} a_p,\label{eq:Ev}
\end{align}
where
\begin{align}
a_p  = & \int_0^{\infty} e^{-\lambda pv}\left(\sum_{m=0}^{M-1} \frac{1}{m!}(\lambda v)^m  \right)^p dv \notag \\
= &\sum_{k_0+\cdots k_{M-1}=p}\binom{p}{k_0,\cdots, k_{M-1}}
\left( \prod_{m=0}^{M-1} \frac{1}{\left(m!\right)^{k_m}}\right) \notag \\
& \hspace{3ex} \times \lambda^{\sum_{m=0}^{M-1}mk_m } \int_0^\infty e^{-\lambda pv} v^{\sum_{m=0}^{M-1} mk_m}dv \label{eq:multinom}\\
=& \frac{1}{\lambda} c_p=\sigma_v^2 c_p, \label{eq:intIdentity}
\end{align}
where \eqref{eq:multinom} follows from the multinomial  expansion theorem, and  \eqref{eq:intIdentity} follows from the integral identity $\int_0^\infty x^n e^{-\mu x}dx=n! \mu^{-n-1}$ \cite{524}. 
The result in \eqref{eq:expOrderstat} for $n=1$ then follows by substituting \eqref{eq:intIdentity} into \eqref{eq:Ev}.

To show \eqref{eq:expOrderstat} for $n\geq 2$, we apply the following recursive relationship obtained from \eqref{eq:Fvn}:
\begin{align}
F_{[n+1]}(v) = F_{[n]}(v) + &  \binom{N_1}{n} F(v)^{N_1-n} \left(1-F(v)\right)^n, \notag \\
 & n=1,\cdots, N_1-1. \label{eq:Fvnplus1}
\end{align}
By substituting \eqref{eq:Fvnplus1} into \eqref{eq:Evn}, we get
\begin{align}
\xE\left[V_{[n+1]}\right]=\xE\left[V_{[n]} \right]- \int_0^\infty  \binom{N_1}{n} F(v)^{N_1-n} \left(1-F(v)\right)^n dv. \label{eq:Evnplus1}
\end{align}
By evaluating the integration in \eqref{eq:Evnplus1}, the result in \eqref{eq:expOrderstat} for $n\geq 2$ can be obtained.

Furthermore, by substituting $M=1$ into \eqref{eq:G1} and with some manipulations, the expression for $G_1(N_1,1)$ as given in \eqref{eq:G1M1} can be obtained.  The property of $G_n(N_1,M)$ given in \eqref{eq:Gorder} can be verified as well. 

This completes the proof of Lemma~\ref{lemma:expV}.
\end{proof}

\section{Proof of Lemma~\ref{lemma:exphnSq}}\label{A:exphnSq}
Note that in the absence of  phase-\rom{1} training ($E_1=0$), the distribution of $\mathbf h_{[n]}$ simply follows from \eqref{eq:hn}, which thus leads to $\xE\left [\left \| \mathbf h_{[n]} \right\|^2 \right]=\beta M$, $\forall n$. In this case, no frequency-diversity gain is achieved.  For the general scenario with $E_1>0$, $[n]$ is the $n$th strongest sub-band out of the $N_1$ trained sub-bands as determined via \eqref{eq:order}. As a consequence, the  statistics of the channel $\mathbf h_{[n]}$ are affected by phase-\rom{1} training via \eqref{eq:ynI} and \eqref{eq:order}. To exploit such a relationship, we first show the following lemma.
\begin{lemma}\label{lemma:statEqv}
The input-output relationship in \eqref{eq:ynI} is statistically equivalent to
\begin{align}\label{eq:hnEqv}
\mathbf h_n = \frac{\beta \sqrt{E_1}}{\beta E_1 + N_0} \mathbf y_n^{\text{\rom{1}}}+ \sqrt{\frac{\beta N_0}{\beta E_1 + N_0}}\mathbf t_n, \ n\in \mathcal{N}_1,
\end{align}
where $\mathbf t_n\sim \mathcal{CN}(\mathbf 0, \mathbf I_M)$ is a CSCG random vector independent of $\mathbf y_n^{\text{\rom{1}}}$, i.e.,
$\xE\left [\mathbf  y_n^{\text{\rom{1}}} \mathbf t_n^H\right]=\mathbf 0, \  n\in \mathcal{N}_1. $
\end{lemma}

\begin{IEEEproof}
It follows from \eqref{eq:hn} and \eqref{eq:ynI} that $\mathbf y_n^{\text{\rom{1}}}$ is a CSCG random vector distributed as
\begin{align}\label{eq:ynIStat}
\mathbf y_n^{\text{\rom{1}}} \sim \mathcal{CN}\left(\mathbf 0, (\beta E_1+N_0) \mathbf I_M\right), \ \forall n.
\end{align}
Furthermore, the cross-correlation between $\mathbf h_n$ and $\mathbf y_n^{\text{\rom{1}}}$ is
\begin{align}\label{eq:crossCorr}
\xE\left[\mathbf y_n^{\text{\rom{1}}} \mathbf h_n^H \right]=\beta \sqrt{E_1} \mathbf I_M.
\end{align}

Since both $\mathbf h_n$ and $\mathbf y_n$ are zero-mean CSCG random vectors, to prove Lemma~\ref{lemma:statEqv}, it is sufficient to show that the random vector $\mathbf h_n$ obtained by \eqref{eq:hnEqv} has the same distribution as \eqref{eq:hn}, and also has the same cross-correlation with $\mathbf y_n^{\text{\rom{1}}}$ as \eqref{eq:crossCorr}. Such results can be easily verified from \eqref{eq:hnEqv} and \eqref{eq:ynIStat}.
\end{IEEEproof}

By applying Lemma~\ref{lemma:statEqv} for the $N_2$ selected sub-bands as determined in \eqref{eq:order}, we obtain the following result,
\begin{align}
\xE\left[ \left\| \mathbf h_{[n]} \right\|^2 \right]&=\frac{\beta^2 E_1}{(\beta E_1 + N_0)^2}\xE\left[ \left \| \mathbf y_{[n]}^{\text{\rom{1}}} \right\|^2\right]+\frac{\beta N_0 M}{\beta E_1 + N_0}\\
&=\frac{\beta^2 E_1 G_n(N_1,M)+\beta N_0 M}{\beta E_1 +N_0}, \ n=1,\cdots, N_2, \label{eq:Ehnsq2}
\end{align}
where \eqref{eq:Ehnsq2} follows from Lemma~\ref{lemma:expV} in Appendix~\ref{A:usefulLemma} and \eqref{eq:ynIStat}.

This completes the proof of Lemma~\ref{lemma:exphnSq}.

\section{Proof of Lemma~\ref{lemma:LMMSE}}\label{A:LMMSE}
Since both $\mathbf h_{[n]}$ and $\mathbf y_{[n]}^{\text{\rom{2}}}$ are zero-mean random vectors with i.i.d. entries, the LMMSE estimator can be expressed as $\hat {\mathbf h}_{[n]}=b_n\mathbf y_{[n]}^{\text{\rom{2}}}$, with $b_n$ a complex-valued parameter to be determined.  
The corresponding MSE can be expressed as
\begin{align}
e_n=& \xE\left[ \left\| \tilde{\mathbf h}_{[n]}\right\|^2\right] =\xE\left[ \left\| \left(1-b_n\sqrt{E_{2,n}} \right) \mathbf h_{[n]} -b_n \mathbf z_{[n]}^{\text{\rom{2}}} \right\|^2\right] \notag \\
=& \left|1-b_n\sqrt{E_{2,n}} \right|^2 R_n(N_1,E_1)  + |b_n|^2 N_0 M \notag \\
=& |b_n|^2 \left( E_{2,n} R_n(N_1,E_1) + N_0 M\right)-(b_n+b_n^*)\sqrt{E_{2,n}} R_n(N_1,E_1)\notag \\
&\hspace{2ex} +R_n(N_1,E_1). \notag
\end{align}
By setting the derivative of $e_n$ with respect to $b_n^*$ equal to zero, the optimal coefficient $b_n$  can be obtained as
\begin{align}
b_n=\frac{\sqrt{E_{2,n}}R_n(N_1,E_1)}{E_{2,n}R_n(N_1,E_1)+N_0M}, \ n=1,\cdots, N_2. \notag
\end{align}
The resulting MMSE can be obtained accordingly.
Furthermore, the following result can be obtained,
\begin{align}
\xE\left[\|\hat {\mathbf h}_{[n]}\|^2 \right]=|b_n|^2 \xE\left[ \|\mathbf y_{[n]}^{\text{\rom{2}}}\|^2\right] & =\frac{E_{2,n} R_n^2(N_1,E_1)}{E_{2,n}R_n(N_1,E_1)+N_0M},\notag \\ & n=1,\cdots, N_2.
\end{align}

To show that $\xE\left[ \tilde {\mathbf h}_{[n]}^H\hat {\mathbf h}_{[n]} \right]=0$, we use the following result,
\begin{align}
\xE\left[ {\mathbf h}_{[n]}^H \hat {\mathbf h}_{[n]}\right]=b \xE\left[ {\mathbf h}_{[n]}^H \mathbf y_{n^{\star}}^{\text{\rom{2}}}\right]& =\frac{E_2 R_h^2(N_1,E_1)}{E_2R_h(N_1,E_1)+N_0M}\notag \\
& =\xE\left[\|\hat {\mathbf h}_{[n]}\|^2 \right].
\end{align}
Therefore, we have
\begin{align}
\xE\left[ \tilde {\mathbf h}_{[n]}^H\hat {\mathbf h}_{[n]} \right]=
\xE\left[ {\mathbf h}_{[n]}^H \hat {\mathbf h}_{[n]}\right]-\xE\left[\|\hat {\mathbf h}_{[n]}\|^2 \right]=0,\ \forall n,
\end{align}
where we have used the identity $\tilde {\mathbf h}_{[n]}={\mathbf h}_{[n]}-\hat {\mathbf h}_{[n]}$.

This completes the proof of Lemma~\ref{lemma:LMMSE}.

\section{Proof of Lemma~\ref{lemma:largeM}}\label{A:largeM}
By applying \eqref{eq:GnMassive} to \eqref{eq:Ehnsq}, we have
$R_n^{\text{large-}M}(N_1,E_1)\rightarrow  \beta M, \ \forall n=1,\cdots, N_2$.
As a result, the average harvested energy at the ER in \eqref{eq:barQ} reduces to
\begin{align}
\bar{Q}^{\text{large-}M}(\{E_{2,n}\}) \rightarrow \eta T P_s \beta \sum_{n=1}^{N_2}  \frac{\beta M E_{2,n}+N_0}{\beta E_{2,n}+N_0},
\end{align}
which is independent of $E_1$ and $N_1$. Therefore, we should have $E_1^{\text{large-}M} \rightarrow 0$, i.e., the phase-\rom{1} training for exploiting the frequency-diversity gain is no longer required.  In this case, $N_2$ out of the $N$ total available sub-bands can be arbitrarily selected  for phase-\rom{2} training in order to apply energy beamforming. Based on \eqref{eq:E2star}, the optimal training energy for each of the $N_2$ sub-bands reduces to
\begin{align}
E_{2,n}^{\text{large-}M} & \rightarrow \left[\sqrt{\eta T P_s(M-1)N_0} -\frac{N_0}{\beta}\right]^+ \notag \\
&\rightarrow \sqrt{\eta T P_sN_0 M}, \ n=1,\cdots, N_2, \text{ as } M\rightarrow \infty. \notag
\end{align}
Furthermore, based on \eqref{eq:QnetN1E1Case2}, the corresponding net harvested energy reduces to
\begin{align}
\Qnet^{\text{large-}M}&\rightarrow N_2\left(\eta T P_s \beta M -2 \sqrt{\eta T P_s N_0 (M-1)}+\frac{N_0}{\beta} \right) \notag \\ 
&\rightarrow \eta T N_2 P_s \beta M, \ \text{ as } M\rightarrow \infty. \notag
\end{align}

This completes the proof of Lemma~\ref{lemma:largeM}.

\section{Proof of Lemma~\ref{lemma:largeNPerfect}}\label{A:largeNPerfect}
To prove Lemma~\ref{lemma:largeNPerfect}, we first show the following result.
\begin{lemma}\label{lemma:GnUB}
For the function $G_n(N_1,M)$ defined in Lemma~\ref{lemma:expV} in Appendix~\ref{A:usefulLemma}, the following inequality holds:
\begin{align}
G_n(N_1,M)\leq M G_1(N_1,1), \forall n=1,\cdots, N_2.
\end{align}
\end{lemma}
\begin{IEEEproof}
Since $G_n(N_1,M)$ monotonically decreases with $n$ as given in \eqref{eq:Gorder}, to prove Lemma~\ref{lemma:GnUB}, it is sufficient to show that $G_1(N_1,M)\leq M G_1(N_1,1)$. To this end, we define $N_1$ i.i.d. zero-mean CSCG random vectors $\mathbf u_i\sim \mathcal{CN}(\mathbf 0, \mathbf I_M)$, $i=1,\cdots, N_1$. Then based on the definition of $G_n(N_1,M)$ given in \eqref{eq:expOrderstat}, we have
\begin{align}
G_1(N_1,M)=&\xE\left[ \underset{i=1,\cdots N_1}{\max} \|\mathbf u_i\|^2\right]
= \xE\left[ \underset{i=1,\cdots N_1}{\max}  \sum_{j=1}^M |u_{ij}|^2\right] \notag \\
\leq & \xE\left[  \sum_{j=1}^M \underset{i=1,\cdots N_1}{\max}  |u_{ij}|^2\right]
=\sum_{j=1}^M \xE \left[ \underset{i=1,\cdots N_1}{\max}  |u_{ij}|^2 \right]\notag \\
=& M G_1(N_1,1),\notag
\end{align}
where $u_{ij}\sim \mathcal {CN}(0,1)$ denotes the $j$th element of $\mathbf u_i$. Note that the last equality follows from the definition of $G_1(N,1)$ and the fact that for any fixed $j$, $\{u_{ij}\}_{i=1}^{N_1}$ are i.i.d. zero-mean unit-variance CSCG random variables.
\end{IEEEproof}

With the expression in \eqref{eq:Qideal1}, the average harvested energy with perfect CSI at the ET is upper-bounded as
\begin{align}
\bar{Q}^{\text{ideal}}(N)&=\eta T P_s \beta \sum_{n=1}^{N_2} G_n (N,M)\leq  \eta T P_s \beta N_2 M G_1(N,1), \label{eq:QidealUB}
\end{align}
where the inequality follows from Lemma~\ref{lemma:GnUB}.
Furthermore, it is easy to see that $\bar{Q}^{\text{ideal}}(N)$ is lower-bounded as
\begin{align}
\bar{Q}^{\text{ideal}}(N)\geq \eta T P_s \beta G_1(N,M) \geq \eta T P_s \beta G_1(N,1),\label{eq:QidealLB}
\end{align}
where the last inequality follows since $G_1(N,M)$ monotonically increases with $M$. It follows from \eqref{eq:QidealUB} and \eqref{eq:QidealLB} that  $\bar{Q}^{\text{ideal}}(N)$ should have the same asymptotic scaling with $N$ as $G_1(N,1)$. Furthermore, for sufficiently large $N$, it follows from \eqref{eq:G1M1} that $G_1(N,1)$ can be expressed as
\begin{align}
G_1(N,1)=\sum_{i=1}^N \frac{1}{i} =\ln N + \gamma + \epsilon_{N},
\end{align}
where $\gamma \approx 0.5772$ is the Euler-Mascheroni constant and $\epsilon \sim \frac{1}{2N}$ which approaches $0$ as $N$ increases. Therefore, for $N\rightarrow \infty$, the following inequalities hold,
\begin{align}
\eta T P_s \beta \ln N \leq  \bar{Q}^{\text{ideal}}(N)  \leq \eta T P_s \beta N_2 M \ln N,
\end{align}
or $\bar{Q}^{\text{ideal}}(N)=\Theta(\ln N)$. This thus completes the proof of Lemma~\ref{lemma:largeNPerfect}.

\section{Proof of Lemma~\ref{lemma:largeNProposed}}\label{A:largeN}
To prove Lemma~\ref{lemma:largeNProposed}, we first derive an upper bound for the net harvested energy $\Qnet(N_1,E_1,\{E_{2n}\})$ of the two-phase training scheme for arbitrary $N_1$, $E_1$, and $\{E_{2,n}\}$. Based on \eqref{eq:Qnet}, we have
\begin{align}
\Qnet & (N_1,E_1,\{E_{2n}\}) \leq  \eta T P_s \sum_{n=1}^{N_2} R_n(N_1,E_1)-E_1N_1 \notag \\
=& \eta T P_s \beta \sum_{n=1}^{N_2} \frac{\beta E_1 G_n(N_1,M)+N_0M}{\beta E_1 + N_0}-E_1N_1 \notag \\
\leq & \eta T P_s \beta N_2 M \frac{\beta E_1 G_1(N_1,1)+N_0}{\beta E_1+N_0}-E_1N_1\notag \\
 \triangleq & \Qnet^{\text{UB}}(N_1,E_1),\notag
\end{align}
where the last inequality follows from Lemma~\ref{lemma:GnUB} in Appendix~\ref{A:largeNPerfect}. Note that the above bound is tight for $M=1$ and $N_2=1$.
As a result, the net harvested energy achieved by the optimized training scheme is upper-bounded by the optimal value of the following  problem.
 \begin{align}
\text{(P1-UB):} \ \underset{ N_1, E_1}{\max}   \quad  & \eta T P_s \beta N_2 M \frac{\beta E_1 G_1(N_1,1)+N_0}{\beta E_1+N_0}-E_1N_1 \notag \\
\text{subject to} \quad &  N_2 \leq N_1 \leq N, \notag \\
& E_1\geq 0. \notag
\end{align}
With $N_1$ fixed, (P1-UB) is a convex optimization problem with respect to $E_1$, whose optimal solution can be expressed as
\begin{align}
E_1^{\text{UB}\star}(N_1)=\left[\sqrt{\frac{\eta T P_s N_2 M N_0(G_1(N_1,1)-1)}{N_1}}-\frac{N_0}{\beta} \right]^+. \label{eq:E1SISO}
\end{align}

The corresponding optimal value of (P1-UB) as a function of $N_1$ can be expressed as
\begin{align}\label{eq:QnetUB}
\Qnet^{\text{UB}}(N_1)=
\begin{cases}
\eta T P_s N_2 M \beta, & \text{ if }  G_1(N_1,1)-1<\frac{N_1}{\Gamma N_2 M} \\
\eta T P_s N_2 M \beta \left(1 + \left( \sqrt{G_1(N_1,1)-1} - \sqrt{\frac{N_1}{\Gamma N_2 M}}\right)^2 \right), & \text{ otherwise}.
\end{cases}
\end{align}
Therefore, problem (P1-UB) reduces to finding the optimal number of training sub-bands as
\begin{equation}
\begin{aligned}
\underset{N_2\leq N_1\leq N}{\max}   \quad  & \Qnet^{\text{UB}}(N_1).\label{P:UB}
\end{aligned}
\end{equation}
It follows from  \eqref{eq:QnetUB} that if $\Gamma$ is small such that $G_1(N_1,1)-1<\frac{N_1}{\Gamma N_2 M}$, $\forall N_1$, $\Qnet^{\text{UB}}$ is independent of $N_1$ and hence problem \eqref{P:UB} has the trivial optimal value $\eta T P_s N_2 M \beta$, which is smaller than the right hand side (RHS) of \eqref{eq:QnetUBProposed}.
On the other hand, for the non-trivial scenario with moderately large $\Gamma$ values, problem \eqref{P:UB} reduces to
\begin{equation}
\begin{aligned}
\underset{N_2 \leq N_1\leq N}{\max}   \quad  &  \sqrt{G_1(N_1,1)-1} - \sqrt{\frac{N_1}{\Gamma N_2 M}}.\label{P:UB2}
\end{aligned}
\end{equation}
Problem \eqref{P:UB2} is a discrete optimization problem due to the integer constraint on $N_1$. To obtain an approximate closed-form solution to \eqref{P:UB2}, we first apply the following approximation for the harmonic sum,
\begin{align}
G_1(N_1,1)=\sum_{i=1}^{N_1}\frac{1}{i} = \ln N_1 + \gamma + \epsilon_{N_1} \approx \ln N_1 +1. \notag
\end{align}
Furthermore, by relaxing the integer constraint and with asymptotically large $N$ and moderate $N_2$ values, problem \eqref{P:UB2} can be approximated as
 \begin{equation}
\begin{aligned}
\underset{\bar N_1}{\max}   \quad  & \sqrt{\ln \bar N_1} - \sqrt{\frac{\bar N_1}{\Gamma N_2 M}}.\label{P:UB3}
\end{aligned}
\end{equation}
By setting the derivative to zero, the optimal solution $\bar N_1^\star$ to problem \eqref{P:UB3} satisfies
\begin{align}
\bar N_1^\star \ln \bar N_1^\star =\Gamma N_2 M, \text{ or } \ln \bar N_1^\star e^{\ln \bar N_1^\star}=\Gamma N_2 M. \label{eq:N1BarStar}
\end{align}
Based on the definition of the Lambert W function, $\bar N_1^\star$ in \eqref{eq:N1BarStar} can be expressed as $\bar N_1^\star=e^{W(\Gamma N_2 M)}$, and the corresponding optimal value of problem \eqref{P:UB} can be obtained as
\begin{align}
\hspace{-2ex} \Qnet^{\text{UB}^\star} \approx \eta T N_2 P_2 \beta M \Big( 1 +\Big(\sqrt{W(\Gamma N_2 M)}-\frac{1}{\sqrt{W(\Gamma N_2 M)}} \Big)^2\Big).
\end{align}

This thus completes the proof of Lemma~\ref{lemma:largeNProposed}.

\bibliographystyle{IEEEtran}
\bibliography{IEEEabrv,IEEEfull}

\end{document}